\documentclass[]{interact}

\usepackage{epstopdf}
\usepackage[caption=false]{subfig}

\usepackage{mathtools}    
\allowdisplaybreaks       
\usepackage{amsmath}

\usepackage[numbers,sort&compress]{natbib}
\usepackage
	{todonotes}

\usepackage{color}

\newcommand{\blu}{\textcolor{blue}}

\newcommand{\bn}{\eta}
\newcommand{\e}{\varepsilon}
\newcommand{\D}{\delta}
\newcommand{\se}{\sqrt{\varepsilon}}
\newcommand{\E}{\mathbb{E}}

\newcommand{\R}{\mathbb{R}}

\newcommand{\CalI}{\mathcal{I}}
\newcommand{\CalD}{\mathcal{D}}

\newcommand{\CalT}{\mathcal{T}}

\newcommand{\CalO}{\mathcal{O}}

\newcommand{\br}{\mathbf{r}}
\newcommand{\bs}{\mathbf{s}}
\newcommand{\bx}{\mathbf{x}}
\newcommand{\by}{\mathbf{y}}
\newcommand{\bu}{\mathbf{u}}
\newcommand{\bz}{\mathbf{z}}
\newcommand{\bk}{\mathbf{k}}

\newcommand{\bq}{\mathbf{q}}

\newcommand{\bv}{\mathbf{v}}

\newcommand{\bX}{\mathbf{X}}

\newcommand{\tbk}{\tilde{\mathbf{k}}}
 \newcommand{\tbq}{\tilde{\mathbf{q}}}

\newcommand{\bcals}{{s}}

\newcommand{\hatr}{\hat{a}_{\rm tr}}

\newcommand{\hbr}{\hat{b}_{\rm ref}}

\newcommand{\hp}{\hat{u}}

\newcommand{\hau}{\hat{a}_{0}}

\newcommand{\Tau}{\tilde{a}}
\newcommand{\Tbu}{\tilde{b}}

\newcommand{\lwu}{\bcals^\e_{0}}

\newcommand{\ba}{\begin{eqnarray}}
\newcommand{\ea}{\end{eqnarray}}
\newcommand{\ban}{\begin{eqnarray*}}
\newcommand{\ean}{\end{eqnarray*}}

\newtheorem{proposition}{Proposition}[section]
\newtheorem{discussion}{Discussion}[section]
\newtheorem{example}{Example}[section]

\newtheorem{assumption}{Assumption}
\theoremstyle{nonumberplain}

\newtheorem*{example31B}{Example \ref{ex1} Revisited} 
\newtheorem*{example11B}{Example \ref{ex:chart}  Revisited} 
\newtheorem{remark}{Remark}[section]

\begin{document}

\articletype{ARTICLE TEMPLATE}

\title{Imaging Through Rough Interfaces: The Shower Curtain Effect}

\author{
\name{Christophe  Gomez\textsuperscript{a}\thanks{CONTACT K.  Sølna. Email: ksolna@uci.edu} and Knut S\o lna\textsuperscript{b}}
\affil{\textsuperscript{a} Aix Marseille Univ, CNRS, I2M, Marseille, France; \textsuperscript{b} UC Irvine, Department of Mathematics, Irvine, CA, USA}
}

\maketitle

\begin{abstract}
The quality of an image observed through a scattering layer, such as a shower curtain, depends strongly on the relative position of the scattering layer between the object and the observer. This well-known phenomenon is commonly referred to as the \emph{shower curtain effect}. When the scattering layer is placed close to the observer, the image is strongly degraded, whereas if it is located close to the object, the object may still be observed with relatively high resolution.

Previous analyses of the shower curtain effect have primarily modeled the scattering layer as a section of a random medium. In this work, we present a new analysis in which the scattering layer is modeled instead as a rough interface, a description that arises naturally in many physical configurations. Within this framework, we derive explicit characterizations of both the image resolution and the signal-to-noise ratio, and determine how these quantities depend on the statistical properties of the rough interface and on its relative location between the object and the observer.
\end{abstract}

\begin{flushleft}
\end{flushleft}


 { \blu{        \tableofcontents  }  }  
 
\section{Introduction}

An intriguing phenomenon in optics is that it is often easier to see a person standing behind a shower curtain than it is for that person to see the observer. This phenomenon is commonly referred to as the \emph{shower curtain effect}. Its origin lies in the interplay between wave scattering and geometry: the quality of the observed image depends not only on the scattering properties of the intervening complex medium, but also on the relative position of that medium between the source and the observer. When the scattering layer is located close to the observer, the image quality deteriorates substantially, whereas when it is located close to the object, the object may still be observed with relatively high resolution.

The shower curtain effect provides a striking example for  how randomness and multiple scattering may influence wave imaging. Beyond its conceptual appeal, the phenomenon is closely connected to important applications in atmospheric optics, remote sensing, biomedical imaging, and imaging through cluttered or turbid environments in general. A central challenge is therefore to develop a precise mathematical description of the effect and to determine how scattering geometry and medium statistics govern image resolution and image stability.

In this paper we present a mathematical analysis of the shower curtain effect in the setting where the scattering layer is modeled as a random rough interface. More precisely, the interface is described as a random two-dimensional surface separating two homogeneous media. This framework is natural for many physical configurations in which scattering is generated predominantly by interface fluctuations rather than by volumetric inhomogeneities.

The present work is in part  motivated by the pioneering contributions of Akira Ishimaru to wave propagation and imaging in random media. His 
monograph \cite{ishimaru78} helped establish a modern theoretical framework for understanding wave scattering in complex environments.  
Moreover, Ishimaru and collaborators were among the first to formulate and analyze the shower curtain effect mathematically in a series of influential papers \citep{yuga86,ishimaru03,ishimaru07}. 

Most previous mathematical studies of the shower curtain effect have been formulated in the context of random volumetric media and build directly on the classical theory of wave propagation in random media \cite{ishimaru78}. By contrast, here we analyze the phenomenon in the setting of random rough interfaces. There is a substantial literature on wave scattering by rough surfaces and interfaces, including perturbative expansions, Kirchhoff approximations, and related asymptotic methods
\cite{Darmon,ishimaru91,TME,Li}. These approaches are physically insightful, but they often rely on restrictive assumptions regarding scaling regimes.
In this work we  consider an asymptotic framework where interfaces are characterized  through their statistical properties and build on the mathematical framework developed in \citep{gomezsolna26} so that 
the scaling regime is precisely   characterized. 
Homogenization methods have also been used to derive effective interface conditions for rapidly varying random interfaces \cite{gallas,nevard}. Our emphasis, however, is on the statistical characterization of the diffuse speckle field generated by the rough interface and on understanding how this diffuse field governs the shower curtain effect. In fact, as we show below, for a homogenized effective flat interface there is no shower curtain effect, highlighting the essential role played by random scattering fluctuations.

An early experimental illustration of the phenomenon can be found in \citep{dror}. In \citep{ishimaru03}, Ishimaru and co-authors investigated the effect in the context of atmospheric cloud scattering and analyzed how the relative location of cloud layers influences imaging performance. Numerical results based on radiative transfer modeling were presented there. In \citep{ishimaru07}, the authors considered a related phenomenon in time-reversal imaging. In that setting, a source emits a wave field that is recorded on a time-reversal mirror, phase conjugated, and re-emitted, leading to refocusing at the original source location. Their numerical study demonstrated a shower curtain effect for time-reversal focusing: the refocusing quality improves significantly when the scattering section is located closer to the source. The analysis there relied in part on paraxial approximations and Gaussian closure assumptions for higher-order wave moments. In the rough-interface setting considered here, the fourth-order statistics can instead be computed more directly, without requiring a Gaussian approximation.

More recently, the shower curtain effect has also been investigated in the context of speckle imaging
\citep{edrei,xie,xie2}. In these works, autocorrelation information extracted from speckle patterns is combined with phase retrieval algorithms to reconstruct hidden objects or mask functions. Related aspects of the effect in active imaging and illumination design are discussed in \citep{tremblay15}. In \citep{thrane,thrane2}, the phenomenon is examined in the context of Optical Coherence Tomography, where the authors argue that imaging performance can be understood and improved in part through the geometric mechanisms underlying the shower curtain effect.
The importance of modeling of  scattering by rough interfaces  
in the context of an imaging problem is also 
discussed in \cite{chrys1,chrys2}.  Here the authors consider and  analyze the challenge of SAR based imaging of
subsurface targets motivated by mine detection.  It is shown how the ground bounce  signal from the 
rough ground surface
presents a challenge for  imaging and how its effect can be mitigated  via data processing tools
involving the singular value decomposition. The authors analyze the importance of the relative magnitude   
of the interface fluctuations  and also discuss the imaging from the point of view of the shower curtain
effect and how the targets placement relative to the rough interface plays a role in the imaging.
Here we only consider randomness and scattering  associated with the rough interface.
 The combined of both volume and interface scattering is discussed in \cite{saba} and it is shown
 that analyzing the joint problem may be important. Analyzing this in the framework considered here 
 is beyond the scope of the present  paper, but will be  considered in future work. 

In the present paper we focus on passive source imaging through a random rough interface and derive explicit characterizations of how the interface statistics and geometric configuration affect imaging performance. A principal objective is to quantify image resolution and to determine precisely how it depends on the interface location and statistical properties. However, high resolution alone is not sufficient for successful imaging if the resulting image is overwhelmed by fluctuations. We therefore analyze not only resolution, but also the signal-to-noise ratio associated with the imaging functionals arising in the optical configuration considered here.

 We next describe a motivating physical experiment. 
        
\begin{example}[USAF Resolution Chart]\label{ex:chart}
A particularly illuminating experimental demonstration of the shower curtain effect is presented in \citep{xie2}. In this experiment, a collimated laser beam illuminates a 1951 USAF resolution chart, and the transmitted wave field is subsequently scattered by a ground-glass diffuser placed between the chart and a detector equipped with an imaging lens. The detector records an image of the resolution chart after propagation through the scattering layer.

A striking observation is that the quality of the recorded image depends strongly on the relative location of the scattering layer. As illustrated in Figure~\ref{fig_setup}, the image quality improves significantly as the scattering medium is moved closer to the resolution chart and farther away from the detector. Thus, even though the scattering properties of the diffuser remain unchanged, the imaging performance varies substantially due solely to the geometric configuration. This sensitivity to relative positioning is the hallmark of the shower curtain effect.

The experiment provides a compelling physical motivation for the present work. It raises a number of natural theoretical questions: What precisely determines the image resolution in the presence of scattering? How do the statistical properties of the scattering interface influence the observed image quality? Why does moving the scattering layer closer to the object improve imaging performance? In this paper, we address these questions within a mathematical framework in which the scattering layer is modeled as a random rough interface. The analysis below shows how the shower curtain effect emerges from the interplay between wave propagation, interface statistics, and imaging geometry.

\begin{figure}
\centering
\subfloat{%
\resizebox*{7cm}{!}{\includegraphics{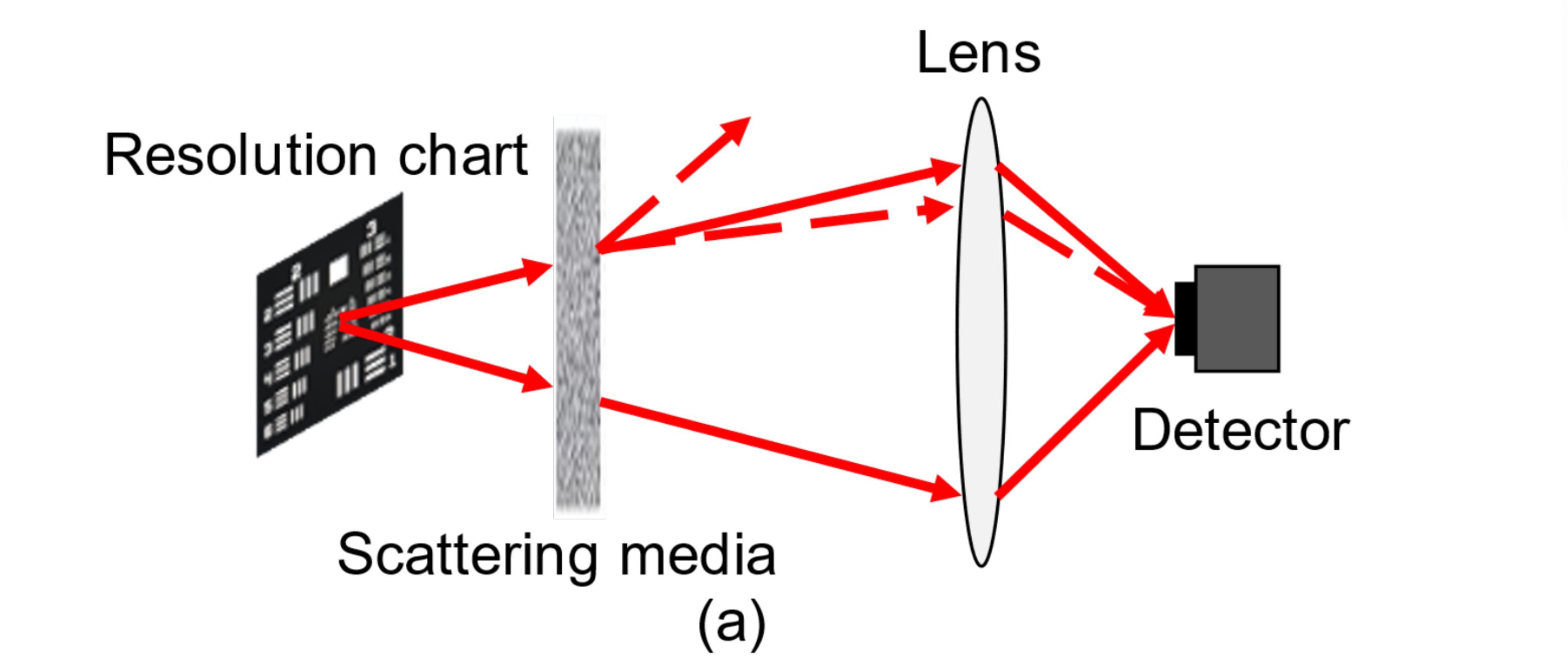}}}\hspace{25pt}
\subfloat{%
\resizebox*{5cm}{!}{\includegraphics{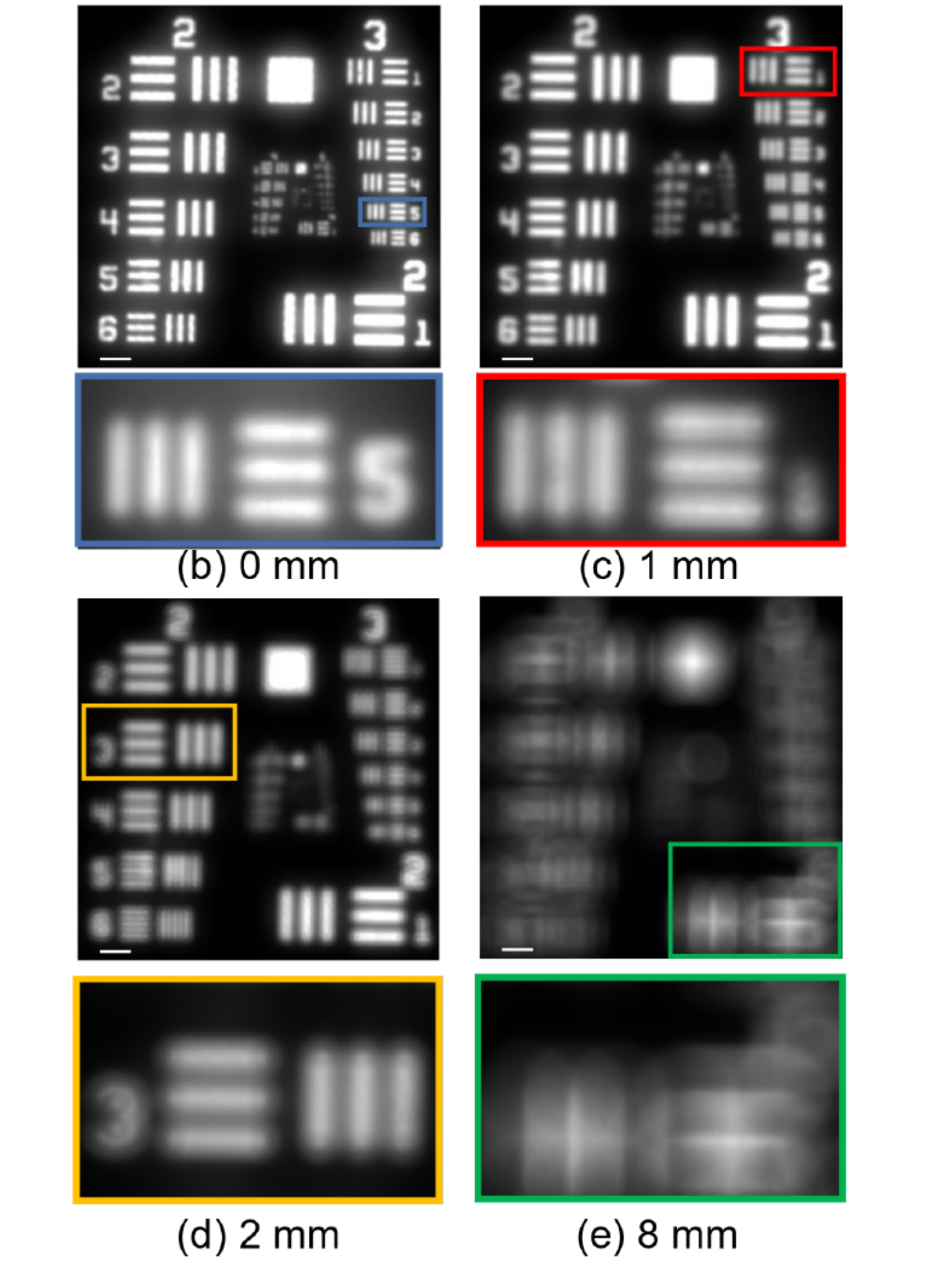}}}\hspace{5pt}

\caption{{\it Figure a)}: Experimental setup with the 1951 USAF resolution chart on the left, a ground-glass scattering layer in the middle, and a detector with an imaging lens on the right. The wave field transmitted through the resolution chart is scattered by the ground glass before reaching the detector. The detector and the resolution chart satisfy a conjugate imaging relationship.
{\it Figure b)--e)} Images recorded for different relative positions of the scattering layer. The distance between the scattering layer and the resolution chart is respectively:
b) 0 mm, c) 1 mm, d) 2 mm, and e) 8 mm (scale bar in the lower left of the subfigures: 400 $\mu$m). The figure illustrates how the image quality deteriorates as the scattering layer is moved away from the object and closer to the detector.
} \label{fig_setup}
\end{figure}
\end{example}

The remainder of the paper is organized as follows. In Section~\ref{sec:beam}, we describe in detail the wave propagation and imaging configuration under consideration and introduce the mathematical framework used throughout the paper. The two imaging modalities analyzed in this work, namely matched field imaging and optical imaging, are then presented in Sections~\ref{sec:match} 
and~\ref{sec:opt}, respectively, where we characterize their resolution and statistical stability properties in the presence of a random rough interface. Finally, in Section~\ref{sec:concl}, we summarize the main conclusions and discuss possible directions for future work.

\section{Beaming Through Rough Interfaces}\label{sec:beam}

We consider three-dimensional linear wave propagation modeled by the scalar wave equation
\begin{equation}\label{eq:wave_equation}
\Delta u - \frac{1}{c^2(\bx,z)}\partial_t^2 u
= F(t,\bx,z),
\qquad
(t,\bx,z)\in \mathbb{R}\times\mathbb{R}^2\times\mathbb{R}.
\end{equation}
Here, the coordinate $z$ represents the principal propagation direction, while $\bx\in\mathbb{R}^2$ denotes the transverse coordinates. The Laplacian decomposes as
\[
\Delta=\Delta_\perp+\partial_z^2,
\]
where $\Delta_\perp$ acts on the transverse variables only. The forcing term $F$ models the source.

We model the medium as two homogeneous half-spaces separated by a random rough interface. Accordingly, the wave speed is given by
\begin{equation}\label{eq:wave_speed}
c(\bx,z)=
\begin{cases}
 c_0 & \text{if }z<z_{\rm int}+\sigma V\Big(\dfrac{\bx}{l_{\rm c}}\Big)
=: z_{\rm int}(\bx),  \\
c_1 & \text{if } z>z_{\rm int}(\bx).
\end{cases}
\end{equation}
The random field $V$ models the interface fluctuations, while $\sigma$ and $l_{\rm c}$ denote respectively the typical amplitude and correlation radius of the interface roughness.

\begin{assumption}\label{ass1}
The random field $V:\mathbb{R}^2\to\mathbb{R}$ is assumed to be mean-zero, stationary, isotropic, bounded, rapidly mixing, and   smooth,  {with bounded second-order derivatives.}
\end{assumption}
We will consider a time harmonic version of \eqref{eq:wave_equation} with  a forcing term of form 
\begin{equation*}
F(t,\bx,z) = e^{-i\omega_{\rm o} t}
\Psi\Big(\frac{\bx}{r_0}\Big)\delta'(z),
\end{equation*}
which models a localized beam emitted from the plane $z=0$ and propagating predominantly in the   $z$-direction. The parameter $\omega_{\rm o}$ denotes the carrier frequency and $r_0$ the transverse beam radius.
The wave field is assumed to satisfy outgoing radiation conditions as $z\to\pm\infty$.

\begin{assumption}
The source profile $\Psi$ is smooth and rapidly decaying.
\end{assumption}

Our objective is to estimate the source profile $\Psi$ from measurements of the transmitted wave field recorded on an observation plane located at
\[
z=z_{\rm obs}>z_{\rm int},
\]
that is, beyond the rough interface; see Figure~\ref{fig1}.

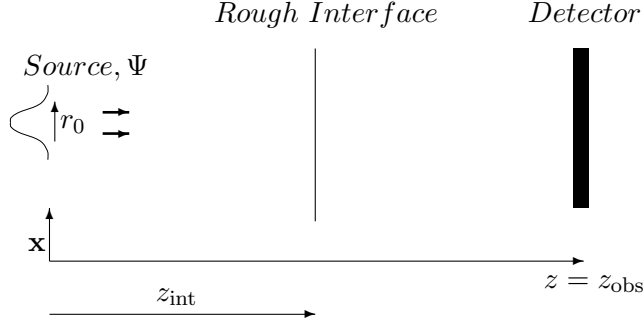
\begin{figure}
\begin{center}
\begin{picture}(300,140)
\linethickness{.35mm}
\put(40,88){\vector(1,0){10}}
\put(40,96){\vector(1,0){10}}
\thinlines
\put(10,110){$Source, \Psi$}
\put(22,85){\vector(0,1){15}}
\put(24,90){$r_0$}
\put(5,78){\includegraphics[width=1cm,angle=90]{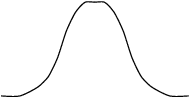}}
\put(82,130){$Rough\ Interface$}
\linethickness{2mm}
\put(220,60){\line(0,1){60}}
\thinlines
\put(200,130){$Detector$}
\put(20,20){\vector(1,0){100}}
\put(60,25){$z_{\rm int}$}
\put(120,55){\line(0,1){65}}
\put(20,40){\vector(1,0){201}}
\put(20,40){\vector(0,1){20}}
\put(12,44){$\bx$}
\put(206,30){$z=z_{\rm obs}$}
\end{picture}
\end{center}
\caption{
Source imaging through a random rough interface. The source is located on the left and the detector on the right.
}
\label{fig1}
\end{figure}

\subsection*{Parameter scaling}

The analysis is based on a separation-of-scales regime. The relevant physical length scales are:
\begin{itemize}
\item the characteristic propagation distance $L_0$,
\item the wavelengths
\[
\lambda_j=\frac{2\pi c_j}{\omega_{\rm o}},
\qquad j=0,1,
\]
\item the beam radius $r_0$,
\item the interface correlation radius $l_{\rm c}$,
\item and the fluctuation amplitude $\sigma$.
\end{itemize}

We assume that both the interface location $z_{\rm int}$ and the observation distance $z_{\rm obs}$ are of order $L_0$.
We consider a high-frequency regime
\[
\e :=
\frac{c_0}{L_0\omega_{\rm o}}
 = 
\frac{\lambda_0}{2\pi L_0}
\ll 1,
\]
  together with a paraxial scaling characterized by
\begin{equation}\label{eq:paraxial}
\frac{r_0^2}{\lambda_0}\sim L_0.
\end{equation}
Equivalently, the Rayleigh length
\[
L_R=\frac{\pi r_0^2}{\lambda_0}
\]
is comparable to the characteristic propagation distance $L_0$. Consequently,
\[
r_0\sim \sqrt{\varepsilon} L_0,
\]
so that
\[
\lambda_0 \ll r_0 \ll L_0.
\]
We further assume a critically scaled rough interface, meaning that the interface fluctuations are of the order of the wavelength:
\begin{equation}\label{eq:sl}
\sigma\sim \lambda_j,
\qquad j=0,1.
\end{equation}

Of particular interest in this paper is the specular scaling regime in which the interface correlation radius is comparable to the beam radius,
\[
l_{\rm c}\sim r_0,
\]
corresponding to the scaling exponent $\gamma=1/2$ introduced below. In this regime the transmitted and reflected fields exhibit random specular effects; see \citep{gomezsolna26}. Another important regime is the speckle regime, in which the interface fluctuations vary on scales much smaller than the beam radius and generate a broad speckle cone.


We adopt the dimensionless scaling
\begin{equation*}
\frac{r_0}{L_0}=\sqrt{\varepsilon},
\qquad
\frac{\sigma}{L_0}=\varepsilon,
\qquad
\frac{l_{\rm c}}{L_0}=\varepsilon^\gamma,
\end{equation*}
which is consistent with \eqref{eq:paraxial} and \eqref{eq:sl}.
Moreover, we introduce dimensionless variables adapted to the paraxial scaling regime by setting
\begin{equation*}
\frac{\bx}{L_0}=\sqrt{\varepsilon}\bx',
\qquad
\frac{z}{L_0}=z',
\qquad
\frac{\omega_{\rm o}L_0}{c_0}=\frac{\omega_{\rm o}'}{\varepsilon},  \qquad
\frac{tL_0}{c_0}=t',
\qquad
\frac{c_i}{c_0}=c_i',
\quad i=0,1.
\end{equation*}
The transverse coordinates are therefore scaled on the beam width scale $\sqrt{\varepsilon}L_0$, while the longitudinal coordinate is scaled on the propagation distance $L_0$. This choice reflects the paraxial geometry in which transverse variations occur on a much smaller scale than longitudinal propagation range. In the following, we drop the primes for notational simplicity and retain the same notation for the corresponding dimensionless variables and parameters. (For notational convenience  we retain the notation $c_0$ despite this parameter being unity
in the dimensionless coordinates.)

\subsection*{The transformed equations}

In view of the time harmonic source term in \eqref{eq:wave_equation}   we make the ansatz  in the dimensionless 
 coordinates
 \begin{equation*}
u(t,\bx,z) =    e^{-i\omega_{\rm o} t/\e}  \,  \check u   (\bx,z)  , 
 \end{equation*}
 and seek a characterization of $ \check u$.   
 From the wave equation \eqref{eq:wave_equation},  we get  the following Helmholtz equation
 in the dimensionless coordinates
 \begin{equation}\label{eq:wave2}
(\partial^2_{zz}  + \e ^{-1} \Delta_{\perp})\check u   (\bx,z) + \frac{\omega^2_{\rm o}}{\e^2 c^2(\bx,z)} \check u   (\bx,z) =  
\Psi(\bx)\delta'(z).
\end{equation}
The dependence on the small parameter $\e$ reflects the high-frequency and paraxial scaling introduced above.
The velocity field $c$ in the dimensionless formulation reads
 \begin{equation*}
c(\bx,z):= \left\{ \begin{array}{ccl} 
c_0 & \text{ if }  &  z<z_{\rm int}(\bx), \\
c_1 & \text{ if }  & z>z_{\rm int}(\bx) .\\
\end{array} \right.
\end{equation*}
for 
 \begin{equation*}
 z_{\rm int}(\bx) = z_{\rm int}
+
\e
V\left(
\frac{\bx}{\e^{\gamma-1/2}}
\right).
\end{equation*}
The dimensionless quantities $z_{\rm int}$ and $z_{\rm obs}$ 
(distances to rough interface and observation plane) are both of order one.

To analyze propagation along the distinguished direction $z$, we Fourier transform with respect to the transverse variables and define
\begin{equation}\label{def:FourierNew}
\hat f(\bk)
:=
\int_{\R^2} 
f(\bx)e^{-i\bk\cdot\bx} d\bx.
\end{equation}
We denote by $\hat u(\bk,z)$ the Fourier transform of $\check u(\bx,z)$ in the transverse variables.
The wave equation \eqref{eq:wave2}, restricted to $z \in (-\infty, z_{\rm int}(\bx)  )$, reduces  then to the 
one-dimensional Helmholtz equation
\begin{equation}\label{eq:Helmholtz}
\partial^2_z \hp(\bk,z) + \frac{\omega_{\rm o}^2}{\e^{2}c_0^2}\Big(1-\e \Big(\frac{c_0}{\omega_{\rm o}}\Big)^{2} |\bk|^2\Big)\hp(\bk,z)=  
\hat{\Psi}\bk)\delta'(z),
\end{equation}
while  for $z \in (z_{\rm int}(\bx),  \infty )$  we get
\begin{equation}\label{eq:Helmholtz1}
\partial^2_{z} \hp(\bk,z) + \frac{\omega_{\rm o}^2}{\e^{2}c_1^2}\Big(1-\e \Big(\frac{c_1}{\omega_{\rm o}}\Big)^{2} |\bk|^2\Big)\hp(\bk,z)= 0.
\end{equation}
 These equations describe the propagation of the transverse Fourier modes through the two homogeneous half-spaces separated by the random rough interface.

\subsection*{Radiation from the source}
The source in   \eqref{eq:wave2} gives wave components  emanating from the source and
propagating in respectively the right- and left-directions.
 Continuity relations across the source plane gives from \eqref{eq:wave2} the following  initial condition 
 for the wave field propagating into the half-space $z>0$  
\begin{equation}\label{eq:IC}
 \hat u_{\rm o}(\bk)  =  \frac{1}{2}\hat\Psi(\bk).
\end{equation}
The derivation is summarized in Appendix~\ref{app:jump}; see also \citep[Section~3.3]{fouque}.

\subsection*{Propagation through the homogeneous medium}

We denote  the vertical slownesses  associated with \eqref{eq:Helmholtz} and \eqref{eq:Helmholtz1} by
 \begin{multline}\label{def_lambda}
  \bcals_j^\e(\bk) := \frac{1}{c_j}
    \sqrt{1 - \e\!\left(\frac{c_j}{\omega_{\rm o}}\right)^{\!2}|\bk|^2}
  = \bcals_j\sqrt{1 - \e\!\left(\frac{c_j}{\omega_{\rm o}}\right)^{\!2}|\bk|^2},\\
  \text{for}\quad \se\,c_j|\bk| < 1,\quad j=0,1.
\end{multline}
where 
\[
\bcals_j:=c_j^{-1} , 
\]
 denotes the background slowness in the homogeneous medium.
The condition in \eqref{def_lambda} ensures that the longitudinal wavenumber remains real and therefore corresponds to propagating modes. When this condition fails, the longitudinal wavenumber becomes purely imaginary and the associated Fourier components are evanescent, decaying exponentially in the propagation direction. Since the observation plane is located at a macroscopic distance from the interface, such modes do not contribute appreciably to the recorded field in the scaling regime considered here. We therefore restrict attention to propagating modes only.
In the forthcoming analysis, the following expansion  will be used:
\begin{equation}\label{eq:exp_lambda}
\frac{1}{\e}\bcals_j^\e(\bk) = \frac{\bcals_j}{\e} -  \frac{c_j}{2  { \omega^2_{\rm o}} }| \bk |^2  + \CalO(\se)\qquad j=0,1,
\end{equation}
where the remainder term is uniform for $\bk$ in the essential support of $\hat\Psi$.

We then  get  for the field radiating from the source and incoming on
the rough interface for $z<\inf_{\bx} z_{\rm int}(\bx)$
 \begin{equation*}
 \hp_{\rm inc}(\bk,z)  :=
\exp\Big(i\frac{\omega_{\rm o}\bcals_0 z}{\e}\Big)
\exp\Big(-i \frac{c_0|\bk|^2 z}{2\omega_{\rm o}}\Big)
\hat u_{\rm o}(\bk),
\end{equation*}
where the first exponential factor represents the rapidly oscillatory carrier wave, while the second factor describes the slower paraxial diffraction of the beam.   
  
The corresponding propagation kernel in physical space is the paraxial Green's function
{\begin{align}\label{eq:par1}
G_j(\omega_{\rm o},\bx,z)
& :=
\frac{1}{(2\pi)^{2}}
\int 
e^{i\bx\cdot\bk}
e^{-i \frac{c_j|\bk|^2 z}{2\omega_{\rm o}}}
d\bk
 =
\frac{\omega_{\rm o}}{2\pi i z c_j}
\exp\left(
i\frac{\omega_{\rm o}|\bx|^2}{2zc_j}
\right),\qquad j=0,1.
\end{align}}
This kernel is the fundamental solution of the paraxial wave equation and characterizes the diffraction-induced spreading of the beam during propagation through the homogeneous medium. 
 
 \subsection*{Transmission through the rough interface}

We next seek a characterization of the scattering  at the rough interface and of the wave field that transmits through it and toward 
the observation plane  $z =  z_{\rm obs} > z_{\rm int}$.  In view of   Assumption
 \ref{ass1} we have assumed a smooth interface with a well defined normal and the wave equation \eqref{eq:wave_equation} yields the following two continuity relations across the randomly perturbed interface:
\begin{equation}\label{eq:continuity_relation}
\check{u}(z=z_{\rm int}(\bx)^+) = \check{u}(z=z_{\rm int}(\bx)^-)\quad\text{and}\quad \partial_z 
\check{u}(z=z_{\rm int}(\bx)^+) = \partial_z \check{u}(z=z_{\rm int}(\bx)^-)  .
\end{equation} 
A discussion of these continuity relations in the present setting is given in Appendix~\ref{app:scatt}.
Our objective is to relate the incoming field at the  {left} of  the interface to the transmitted field to the  {right} of it. 
To this end, we decompose the wave field into forward- and backward-propagating components and use the transmission conditions \eqref{eq:continuity_relation} to derive the corresponding coupling relations across the rough interface. The derivation follows the general framework developed in \cite{gomezsolna26}; for completeness, the main steps relevant to the present scaling regime are summarized in Appendix~\ref{app:iscatt} and we discuss next the 
result that follows from this analysis. 

Let
\[
z^- :=\inf_{\bx} z_{\rm int}(\bx),
\qquad
z^+ :=\sup_{\bx} z_{\rm int}(\bx),
\]
denote planes located respectively  {on the left} and  {on the right of} the rough interface. We further introduce the effective transmission coefficient associated with the corresponding flat interface,
\begin{equation}\label{def}
\CalT
:=
\frac{2\bcals_0}{\bcals_0+\bcals_1} = \frac{2c_0}{c_0+c_1},
\end{equation}
where $\bcals_j=c_j^{-1}$, $j=0,1$, are the background slownesses in the two homogeneous half-spaces.

At leading order in the paraxial scaling regime, the transmitted field to the left of  the interface is then given by
\begin{equation}\label{eq:trans}
\hat u(\bk,z^+)
\simeq
\CalT
e^{i\omega_{\rm o}
\left[
\bcals_0(z_{\rm int}-z^-)
+
\bcals_1(z^+-z_{\rm int})
\right]/\e}
\int 
\mathfrak{K}^\e\left(
\omega_{\rm o}(\bcals_0-\bcals_1),
\bk,\bk'
\right)
\hat u_{\rm inc}(\bk',z^-) d\bk', 
\end{equation}
where we introduced the scattering operator
\begin{equation}\label{eq:K}
\mathfrak{K}^\e (\tau,\bk,\bk') := \ \frac{1}{(2\pi)^2}
\int 
e^{i(\bk'-\bk)\cdot\bx}
e^{i\tau
V\left(
\bx/\e^{\gamma-1/2}
\right)}
 d\bx.
\end{equation}
The oscillatory factor in \eqref{eq:K} shows that the random interface redistributes energy among the transverse Fourier modes. The statistical properties of the transmitted field are therefore determined by the interplay between the roughness scaling, encoded through the parameter $\gamma$, and the fluctuations of the random field $V$.

\subsection*{The wave field in the observation plane}

We introduce the notation
\begin{equation}\label{eq:hG}
\hat G_j(\omega,\bk,z)
:=
\exp\left(
-i\frac{z c_j |\bk|^2}{2\omega}
\right),
\qquad j=0,1,
\end{equation}
for the Fourier representation of the paraxial propagator at frequency $\omega$ in the two homogeneous half-spaces separated by the rough interface. Equivalently, $\hat G_j$ is the Fourier transform in the transverse variables of the semigroup associated with the paraxial Schr\"odinger equation governing beam propagation in the homogeneous medium.

Combining the propagation through the lower homogeneous medium, the scattering relation \eqref{eq:trans}, and the subsequent propagation to the observation plane, we obtain the following leading-order expression for the transmitted field at
$
z=z_{\rm obs}:
$
\begin{align}\nonumber 
 \check u^\e(\bx,z_{\rm obs}) &
 \simeq
\frac{\CalT}{(2\pi)^2}
e^{i\omega_{\rm o}
\left[
\bcals_0 z_{\rm int}
+
\bcals_1(z_{\rm obs}-z_{\rm int})
\right]/\e}
\iint 
e^{i\bk\cdot\bx}
\hat G_1(\omega_{\rm o},\bk,z_{\rm obs}-z_{\rm int})
  \\
&\quad \times
\mathfrak K^\e\left(
\omega_{\rm o}(\bcals_0-\bcals_1),
\bk,\bk'
\right)
\hat G_0(\omega_{\rm o},\bk',z_{\rm int})
\hat u_{\rm o}(\bk')
 d\bk' d\bk, \label{eq:wobs}
\end{align}
where 
the transmission coefficient $\CalT$ is given by \eqref{def}, the scattering operator $\mathfrak K^\e$ is defined in \eqref{eq:K}, and the incoming source field is given by \eqref{eq:IC}.

Expression \eqref{eq:wobs} has a natural interpretation: the source field first propagates paraxially from the source plane to the rough interface through the medium with wave speed $c_0$, is subsequently scattered by the random interface through the mode-coupling operator $\mathfrak K^\e$, and finally propagates to the observation plane through the upper homogeneous medium with wave speed $c_1$.

The phase factor
\[
e^{i\omega_{\rm o}
\left[
\bcals_0 z_{\rm int}
+
\bcals_1(z_{\rm obs}-z_{\rm int})
\right]/\e}
\]
is spatially constant on the observation plane and  does not influence intensity-based observables. Similarly, the coefficient $\CalT$ only produces a deterministic amplitude scaling. For this reason, it is convenient to introduce the phase- and transmission-compensated field
\begin{equation*}
\frac{2}{\CalT}
e^{-i\omega_{\rm o}
\left[
\bcals_0 z_{\rm int}
+
\bcals_1(z_{\rm obs}-z_{\rm int})
\right]/\e}
\check u^\e(\bx,z_{\rm obs}),
\end{equation*}
and define the corresponding effective observation model
\begin{align}\nonumber
U_{\rm obs}^\e(\bx)
&=
\frac{1}{(2\pi)^2}
 \iint 
e^{i\bk\cdot\bx}
\mathfrak K^\e\left(
\omega_{\rm o}(\bcals_0-\bcals_1),
\bk,\bk'
\right) \\
 &\qquad \qquad \times
\hat G_1(\omega_{\rm o},\bk,z_{\rm obs}-z_{\rm int})
\hat G_0(\omega_{\rm o},\bk',z_{\rm int})
\hat\Psi(\bk')
 d\bk' d\bk. \label{eq:wobs2b}
\end{align}
This representation forms the basis for the subsequent analysis of imaging resolution and statistical stability in the presence of the random rough interface.  
    
\section{Matched Field  Imaging }\label{sec:match}

We first consider imaging based on the full transmitted wave field.
Our objective is to reconstruct the source profile $\Psi$ from measurements of the transmitted field recorded in the observation plane $z=z_{\rm obs}$.
Throughout this section, we assume that the background parameters
$c_j$, $j=0,1$, together with the interface location $z_{\rm int}$ and the observation distance $z_{\rm obs}$, 
are known and available for constructing the imaging functional.
We also assume that the receiver aperture is sufficiently large to capture the transmitted wave field without truncation effects.

In order to identify the matched field imaging functional assume a point source at $\by$ so that 
$\Psi(\bx)  =\delta(\bx-\by)$, moreover,  a flat interface so that 
\[
\mathfrak{K}^\e({\omega_{\rm o}}(\bcals_0-\bcals_1),  \bk, \bk') =\delta(\bk'-\bk).
\]
The corresponding transmitted field at the observation plane is then
\begin{align}\label{eq:wobs3}
U_{\rm mf}(\bx;\by)
&=
\frac{1}{(2\pi)^2}
\int 
\hat G_1(\omega_{\rm o},\bk,z_{\rm obs}-z_{\rm int})
\hat G_0(\omega_{\rm o},\bk,z_{\rm int})
e^{i\bk\cdot(\bx-\by)}
d\bk,\nonumber \\
&=
\int 
G_1(\omega_{\rm o},\bx-\bx',z_{\rm obs}-z_{\rm int})
G_0(\omega_{\rm o},\bx'-\by,z_{\rm int})
d\bx'.
\end{align}
This field serves as the reference template in the matched field imaging procedure
defined by
\begin{align*}
\CalI_{\rm mf}(\by)  =
\int 
U_{\rm obs}^{\e}(\bx)
\overline{U_{\rm mf}(\bx;\by)}
 d\bx ,
\end{align*}
where the overline denotes complex conjugation.

To characterize the effect of the random interface, we introduce the characteristic function
\begin{equation}\label{def:charac_V}
\phi_V(\tau)
:=
\E\big[
e^{i\tau V({\bf 0})}
\big],
\end{equation}
associated with the interface fluctuation process $V$.

The following result describes the mean matched field image.
 \begin{proposition}\label{prop1}
 \begin{equation*}
\E\left[ \CalI_{\rm mf}(\by) \right]  = \phi_V(\omega_{\rm o}(\bcals_0 - \bcals_1) ) \Psi(\by).
\end{equation*}
\end{proposition}

\begin{discussion}
In the absence of interface fluctuations, that is when the interface is flat and $V\equiv 0$, we have
\[
\phi_V(\tau)\equiv 1,
\]
and therefore
\[
\CalI_{\rm mf}(\by)=\Psi(\by).
\]
Thus, under ideal homogeneous propagation conditions, the matched field functional reconstructs the source profile exactly.

We introduce the differential  wave number associated with the
medium contrast
\begin{align}\label{eq:dkdef}
   \Delta k & =  \omega_{\rm o}|\bcals_0 - \bcals_1|  = \frac{\omega_{\rm o}}{c_0}\Big\vert 1 - \frac{c_0}{c_1} \Big\vert  
\end{align}
and  observe that the mean image is significantly attenuated when
\[
\sigma_V \Delta k \gtrsim 1,
\]
where $\sigma_V$ denotes the characteristic amplitude of the (rescaled) interface fluctuations.
This regime corresponds  to large interface fluctuations and or to strong contrast between the two media.
In the special case when $V$ has the Gaussian distribution we have  explicitly 
\[
    \phi_V(\omega_{\rm o}(\bcals_0 - \bcals_1) )   = \exp\left(  - \frac{(\sigma_V \Delta k )^2}{2}   \right) . 
\]

An important observation is that the mean imaging functional does not depend on the interface location $z_{\rm int}$.
Consequently, at the level of the mean matched field image, no shower curtain effect is observed.
Likewise, the spatial correlation length of the interface fluctuations does not influence the mean image.
As shown below in Proposition~\ref{prop2}, these effects instead enter through the variance and statistical stability properties of the imaging functional.
\end{discussion}

\begin{proof}
We use the shorthand notation 
\begin{align}
\hat{G}_0(\bk)    = \hat{G}_0(\omega_{\rm o},\bk,z_{\rm int}) , \quad 
\hat{G}_1(\bk)   = \hat{G}_1(\omega_{\rm o},\bk,z_{ \rm obs}-z_{\rm int}) , \label{eq:sh}
\end{align}
 then in view of  \eqref{eq:wobs2b}  and using the notation   \eqref{eq:dkdef} the mean recorded field then reads 
\[
\E\left[ U_{\rm obs}^{\e}(\bx )  \right]   = 
 \frac{1}{(2\pi)^2  }
    \iint    e^{i \bk \cdot \bx } 
   \E\left[\mathfrak{K}^\e(\Delta k,  \bk, \bk')  \right]
\hat{G}_1( \bk ) \hat{G}_0( \bk' )
\hat\Psi(\bk') d\bk' d\bk  .
\]
Due to the stationarity of the fluctuations $V$, we  have  in the sense of distributions 
\[
\E\left[\mathfrak{K}^\e(\Delta k,  \bk, \bk')  \right] = 
\phi_V\big(
\Delta k
\big)
\delta(\bk-\bk').
\]
and hence 
\begin{align*}
\E \left[U_{\rm obs}^{\e}(\bx)\right]
&=
\phi_V\big(
\Delta k
\big)
\frac{1}{(2\pi)^2}
\int 
e^{i\bk\cdot\bx}
\hat G_1(  \bk )
  \hat G_0( \bk )
\hat\Psi(\bk)
d\bk .
\end{align*}
Transforming back to the spatial domain yields
\begin{align}\label{eq:meani}
\E\left[U_{\rm obs}^{\e}(\bx)\right]
&=
\phi_V\big(
\Delta k
\big)
\iint 
G_1( \bx-\bx' )
 G_0( \bx'-\by )
\Psi(\by)
d\by d\bx'.
\end{align}

Note next  the time-reversal refocussing (reciprocity) relation 
\begin{align}\label{eq:tr}
  \int  G_j( \bx-\bx')\,
    & \overline{G_j( \bx'-\by )}\,d\bx' \nonumber\\
  &= \frac{k_j^2}{(2\pi z)^2}
    \int  \exp\!\left(
      \frac{ik_j}{2z}\bigl(|\bx-\bx'|^2 - |\bx'-\by|^2\bigr)
    \right)d\bx' \nonumber\\
  & = \frac{k_j^2}{(2\pi z)^2}
    \int  \exp\!\left(
      \frac{ik_j}{z}\bigl(|\bx|^2 - |\by|^2 + \bx'\cdot(\by-\bx)\bigr)
    \right)d\bx' 
    \nonumber\\
  & 
    = 
\delta(\bx-\by),
\end{align}
for $j=0,1$. Using  the expression for $\E[U_{\rm obs}^{\e}]$ in  \eqref{eq:meani} and applying \eqref{eq:tr} twice gives

\begin{equation}\label{eq:wobs4}
\E\left[\CalI_{\rm mf}(\by)\right]  =
\int 
\E\left[ U_{\rm obs}^{\e}(\bx) \right]
\overline{U_{\rm mf}(\bx;\by)}
 d\bx  
 =
\phi_V(
\Delta k
)
\Psi(\by),
\end{equation}
which completes the proof.
\end{proof}

To assess the statistical stability of the matched field imaging functional, we now analyze its variance.
While Proposition~\ref{prop1} shows that the mean image is simply a damped version of the source profile, the variance reveals how fluctuations of the rough interface affect the variations  of the image from one realization of the interface to another.
In order to state the result it will be convenient to   introduce  the characteristic function for the differential form of the rough
interface fluctuations
 \begin{equation*}
\phi_{\Delta V}(\tau, \bx) :=
\E\big[
e^{i\tau (V(\bx+\bx')-V(\bx'))}
\big].
\end{equation*}
By stationarity and isotropy of the random field $V$, the function
$\phi_{\Delta V}$  depends only on $|\bx|$ and is real-valued and even in $\bx$.

The following proposition identifies three qualitatively different scattering regimes according to the value of the scaling exponent $\gamma$.
These correspond respectively to critically scaled ($\gamma=1/2$), slowly varying ($\gamma<1/2$), and rapidly varying ($\gamma>1/2$) interface fluctuations relative to the transverse beam width.

\begin{proposition}\label{prop2}
The variance of the image is :
\begin{itemize}
\item for $\gamma=1/2$
\begin{equation}\label{eq:v1}
 {\rm Var}\bigl(\CalI_{\rm mf}(\by)\bigr) = \int\kappa(\Delta k,\bk)
    \left|\Psi\!\Big(\by + \frac{c_0 z_{\rm int}}{\omega_{\rm o}}\bk\Big)\right|^2
    d\bk - \bigl|\phi_V(\Delta k)\,\Psi(\by)\bigr|^2;
\end{equation}
\item for $\gamma<1/2$ and $\e\to 0$
\[
   {\rm Var}\bigl(\CalI_{\rm mf}(\by)\bigr)
  = |\Psi(\by)|^2\bigl(1-\bigl|\phi_V(\Delta k)\bigr|^2\bigr);
\]
\item  for $\gamma > 1/2$ and $\e\to 0$
  \[ 
  {\rm Var}\bigl(\CalI_{\rm mf}(\by)\bigr) = 0,
\] 
 \end{itemize}
with $\Delta k$ defined in \eqref{eq:dkdef} 
and the scattering kernel $\kappa$ defined
by
 \begin{align}\label{eq:kapdef}
   \kappa(\tau,\bk) =   \frac{1}{(2\pi)^2} \int    {e^{-i \br\cdot \bk}} \phi_{\Delta V}(\tau, \br)
   d\br    
   .
\end{align}
 \end{proposition}
 
 
 \begin{discussion}\,
 \begin{itemize}
 \item \emph{Critical regime ($\gamma=1/2$).} 
The interface fluctuations interact with the beam on the same transverse scale as the beam width. The resulting image fluctuations are described by the scattering kernel $\kappa$, which acts as an effective point-spread function. The image variance therefore reflects both the roughness statistics of the interface and the propagation distance $z_{\rm int}$. 
 This is precisely the regime in which the shower-curtain effect becomes observable.
  Note  that for a flat interface $\kappa(\tau,\bk) = \delta(\bk)$ and $\phi_V(\tau)\equiv 1$
 so that indeed ${\rm Var}(\CalI_{\rm mf}(\by)) = 0$. 
 
\item \emph{Large-scale interface fluctuations ($\gamma<1/2$).}  
  The rough interface fluctuations takes place on a scale that is large compared to the beam width,
  so over the beam footprint the random phase is approximately constant.  Thus, imaging of the source is 
 possible by considering the  magnitude of the imaging function.  The variance of the image function itself will be 
 of relative magnitude ${\mathcal O}(1)$ in general due to the phase factor.
   
\item  \emph{Small-scale interface fluctuations ($\gamma>1/2$).}
 The interface fluctuations happens on a scale that is fast relative to  the beam width so that there is an 
 `averaging  effect'  over the beam width regarding the impact of the rough interface.  From the point of view of the specular  cone then
  the rough interface can be modeled by a flat homogenized interface \cite{gomezsolna26}. This is the speckle regime and the transmitted random wave  fluctuations 
 will be small,  but supported in a relatively  wide cone, the  incoherent  image fluctuations will accordingly  be small.
 However, if   $\sigma_V  \Delta k \gtrsim 1 $ the image will be strongly damped and hence sensitive 
 to additive measurement noise. 
\end{itemize}
   
   \end{discussion} 
\begin{proof}
Consider  the second-order moment of the imaging function 
\begin{equation}\label{eq:mfd}
\E\left[ |\CalI_{\rm mf}(\by)|^2 \right]    =  \iint \E\big[ U_{\rm obs}^{\e}(\bx ) \overline{U_{\rm obs}^{\e}(\bx' ) } \big]  
\overline{U_{\rm mf}(\bx; \by)}   U_{\rm mf}(\bx'; \by)  d\bx d\bx' .
\end{equation} 
Observe  that  
\begin{align*}
 \E\left[\mathfrak{K}^\e(\tau,  \bk, \bk') \overline{\mathfrak{K}^\e(\tau,  \bq, \bq')} \right] 
 = \frac{\delta((\bk'-\bq')-(\bk-\bq))}{(2\pi)^2} \iint  {e^{-i(\bq-\bq')\cdot \br}}  { \phi_{\Delta V}\left(\tau, \frac{\br}{\e^{\gamma-1/2}}\right) } d\br.
\end{align*}
For the cross-moment of the observations we then have from \eqref{eq:wobs2b}, and using again the notation \eqref{eq:sh}
 \begin{align}\label{eq:2ndmom}
  \E\left[U_{\rm obs}^{\e}(\bx)\,\overline{U_{\rm obs}^{\e}(\bx')}\right]
    = \frac{1}{(2\pi)^4}&\iiiint e^{i(\bk\cdot\bx - \bq\cdot\bx')}      \E\Bigl[
      \mathfrak{K}^\e(\Delta k,\bk,\bk')\,
      \overline{\mathfrak{K}^\e(\Delta k,\bq,\bq')}
    \Bigr] \nonumber\\
  &\times
    \hat{G}_1( \bk )\,
    \hat{G}_0( \bk' )      \overline{\hat{G}_1( \bq )}\,
    \overline{\hat{G}_0( \bq' )}\,
    \hat\Psi(\bk')\overline{\hat\Psi(\bq')}\,
    d\bk'd\bk\,d\bq'd\bq  \nonumber 
%
\\ 
 = \frac{1}{(2\pi)^6  } &
    \iiiint    e^{i (\bk \cdot \bx - \bq \cdot \bx') }  e^{-i  \br \cdot (\bq-\bq')      } 
    { \phi_{\Delta V}\left(\Delta k ; \frac{\br}{\e^{\gamma-1/2}} \right) } \nonumber 
   \\ & \hspace{1cm} \times 
\hat{G}_1( \bk ) \hat{G}_0( \bk-(\bq-\bq') )
\overline{\hat{G}_1( \bq )} \,  \overline{\hat{G}_0( \bq' )} \\
& \hspace{1cm} \times \hat\Psi(\bk - (\bq-\bq')) \overline{\hat\Psi(\bq')}   d\bk d\bq' d\bq d\br \nonumber.
\end{align}
Then in view of \eqref{eq:wobs3}  and   \eqref{eq:mfd},  and using \eqref{eq:hG} we have 
\begin{align}
\E\left[ |\CalI_{\rm mf}(\by)|^2 \right] 
=\frac{1}{(2\pi)^{10}  } & \int \dots\int    e^{i (\bk \cdot \bx - \bq \cdot \bx')} {e^{-i \br\cdot(\bq-\bq')}} \phi_{\Delta V}\Big(\Delta k ; \frac{\br}{\e^{\gamma-1/2}} \Big) 
 \nonumber   \\ & \hbox{} \times \hat{G}_1( \bk )  \hat{G}_0( \bk-(\bq-\bq') )
\overline{\hat{G}_1( \bq )} \,  \overline{\hat{G}_0( \bq' )} 
 \nonumber \\ & \hbox{} \times 
 \overline{\hat{G}_1( \bs )} \, \overline{ \hat{G}_0( \bs )}
  \hat{G}_1( \bs' )  \hat{G}_0( \bs') \, e^{-i \bs\cdot(\bx-\by)} \, e^{i \bs'\cdot(\bx'-\by)} \nonumber \\
&\times \hat\Psi(\bk-(\bq-\bq')) \overline{\hat\Psi(\bq')}  d\bk d\bq' d\bq d\br  d\bs d\bs' d\bx d\bx' \nonumber   \\ 
= \frac{1}{(2\pi)^6} & \iiiint e^{i \by \cdot (\bk - \bq) } { e^{-i \br\cdot(\bq-\bq')}} \phi_{\Delta V}\Big(\Delta k ; \frac{\br}{\e^{\gamma-1/2}} \Big)  \nonumber \\ 
& \hbox{} \times \hat{G}_0( \bk-(\bq-\bq') ) \overline{\hat{G}_0( \bq' )} \, \overline{ \hat{G}_0( \bk )} \hat{G}_0( \bq )  \label{eq:varE}   \\
& \times \hat\Psi(\bk-(\bq-\bq')) \overline{\hat\Psi(\bq')} d\bk d\bq' d\bq d\br \nonumber \\ 
= \frac{1}{(2\pi)^6} & \iiiint e^{i \by \cdot (\bk - \bq) }  {e^{-i \br\cdot(\bq-\bq')}} \phi_{\Delta V}\Big(\Delta k ; \frac{\br}{\e^{\gamma-1/2}} \Big)      
 \nonumber    \\ & \hbox{} \times 
   e^{-i \frac{ z_{\rm int} c_0}{\omega_{\rm o}}   (\bq - \bk)\cdot(\bq-\bq') }
  \hat\Psi(\bk - (\bq-\bq') ) \overline{\hat\Psi(\bq')}   d\bk d\bq' d\bq d\br.    \label{eq:mf3}
\end{align}

Consider first  the case $\gamma=1/2$. We  find after   integrating in $\bk$ and $\br$  
\begin{align*} 
\E\left[ |\CalI_{\rm mf}(\by)|^2 \right]    =  
 \frac{1}{(2\pi)^2}
    \iint  &  e^{- i \by \cdot \bq' } \kappa(\Delta k,  \bq-\bq') e^{-i \frac{ z_{\rm int} c_0}{\omega_{\rm o}}  \bq' \cdot (\bq-\bq' ) }     
 \nonumber    \\ 
 &  \times  \Psi\Big(\by +  \frac{ z_{\rm int} c_0}{\omega_{\rm o}} (\bq-\bq') \Big) \overline{\hat\Psi(\bq')} d\bq' d\bq.
\end{align*} 
Then   by denoting $\bk=\bq-\bq'$ and integrating in $\bq'$ we find 
\begin{align*} 
& \E\left[ |\CalI_{\rm mf}(\by)|^2 \right]    =  
     \iint        \kappa(\Delta k,  \bk ) 
     \left|  \Psi\Big(\by +  \frac{ z_{\rm int} c_0}{\omega_{\rm o}} \bk  \Big)  \right|^2
 d\bk     .
\end{align*} 
 Finally, after subtracting {$|\E[\CalI_{\rm mf}(\by)]|^2$, given by \eqref{eq:wobs4}}, we obtain \eqref{eq:v1}.

 Consider next the case $\gamma < 1/2$.
  We then have 
     \begin{align*} 
  \lim_{\e \to 0}  \phi_{\Delta V}\Big(\Delta k ; \frac{\br}{\e^{\gamma-1/2}} \Big) 
  =\phi_{\Delta V}(\Delta k ;  {\bf 0} )
  =  1 . 
    \end{align*}
    Integration in $\br$ next produces a $\delta(\bq-\bq')$ and we find subsequently 
   from  \eqref{eq:mf3} 
  \[
   \E\left[ |\CalI_{\rm mf}(\by)|^2 \right]    =  |\Psi(\by)|^2 ,
   \]
 and the result in this case follows. 
    
 Finally, consider    the case $\gamma > 1/2$.   In view of Assumption \ref{ass1} and the mixing property of  $V$ we then have 
for $\br \neq {\bf 0}$ 
   \begin{align*} 
  \lim_{\e \to 0} \phi_{\Delta V}\Big(\Delta k ; \frac{\br}{\e^{\gamma-1/2}} \Big) 
     =  |\phi_{ V}(\Delta k)|^2  . 
  \end{align*}
  Then we get   from  \eqref{eq:mf3}.  
  \[
   \E\left[ |\CalI_{\rm mf}(\by)|^2 \right]    =  |\phi_{ V}\left(\Delta k\right) \Psi(\by)|^2 = 
   \E^2\left[  \CalI_{\rm mf}(\by)  \right] ,
   \]
so that the variance goes to zero  in this limit.

     \end{proof}

\begin{remark}\label{remark1}
Under Assumption~\ref{ass1}, the covariance function of the interface fluctuations is
\[
C(|\bx|)
= \E\left[V(\bx')V(\bx'+\bx)\right].
\]
A natural measure of the local roughness of the interface is
\begin{equation}\label{eq:Ddef}
\CalD
:=
-\left.\frac{d^2}{dx^2}C(x)\right|_{x=0},
\end{equation}
which quantifies the typical squared slope of the interface fluctuations.
Note that $\CalD$ can be interpreted as the mean square slope of the interface 
fluctuations, and thus a measure of their roughness.  

We introduce next the following blurring parameter
\begin{equation}\label{def:Bm}
    \beta_{\rm int}   = z_{\rm int}
\left|1-\frac{c_0}{c_1}\right|
\sqrt{\CalD}.
\end{equation}
This parameter plays the role of an effective blurring radius associated with the rough interface. It increases linearly with the distance from the source plane to the interface and with the roughness level of the interface. As will be seen both here and in the optical-imaging setting considered later, $\beta_{\rm int}$  is the fundamental quantity controlling the broadening of incoherent image fluctuations and therefore provides a quantitative measure of the shower-curtain effect.
\end{remark}

\begin{example}\label{ex1}
We illustrate Proposition~\ref{prop2} in the critical regime
$
\gamma=1/2,
$
where the correlation length of the interface fluctuations is comparable to the beam width.

For analytical convenience, we model the interface fluctuations by a stationary Gaussian random field and choose a Gaussian source profile,
\begin{align}\label{eq:defSV}
V(\bx)
&\sim {\mathcal N}(0,\sigma_V^2), \nonumber \\
C(|\bx|)
&= \E[V(\bx')V(\bx'+\bx)]
= \sigma_V^2
\rho_V\left(\frac{|\bx|}{\ell_V}\right),
\nonumber \\
\Psi(\bx)
&= \sigma_\Psi
\exp\left(
-\frac{|\bx|^2}{{2}\ell_\Psi^2}
\right).
\end{align}
Here $\ell_V$  denotes the correlation length of the interface fluctuations and $\ell_\Psi$  the beam width in the transverse scaling. We then obtain
 \begin{align*}
  \phi_{ V}\left(\Delta k \right) & =  \exp\left(
-\frac{(\sigma_V\Delta k)^2}{2} \right) 
 , \\
 \phi_{\Delta V}\left(\Delta k ;  {\bx}  \right) &=   \exp\left(
-(\sigma_V\Delta k)^2
\left[
1-\rho_V\left(\frac{|\bx|}{\ell_V}\right)
\right]
\right).
\end{align*}
Substituting these expressions into Proposition~\ref{prop2} yields
\begin{align}
{\rm Var}(\CalI_{\rm mf}(\by))
&=
\frac{1}{(2\pi)^2}
\iint 
 { e^{-i\bk\cdot\br}}
e^{-(\sigma_V\Delta k)^2}
\Bigl(
e^{(\sigma_V\Delta k)^2
\rho_V(|\br|/\ell_V)}
-1
\Bigr)
\nonumber \\
&\qquad\qquad\times
\left|
\Psi\left(
\by+\frac{c_0 z_{\rm int}}{\omega_o}\bk
\right)
\right|^2
d\bk d\br ,
\label{eq:varmf}
\end{align}
for $\Delta k = \omega_{\rm o}|\bcals_0 - \bcals_1| $ the differential wave number associated  with the medium contrast  
   introduced  in \eqref{eq:dkdef}.
    Note that   we have 
   \begin{equation*}
     \CalD  = - \frac{d^2}{d x^2 }  C(x) _{|x=0}    = D  \Big( \frac{\sigma_V}{\ell_V} \Big)^2. 
   \end{equation*} 
   where 
   \[
   D:=-\rho_V''(0)
   \]  is a dimensionless measure of the shape of the covariance function near the origin.
 The    blurring parameter introduced in Remark \ref{remark1}  
 can then be written as  
 \begin{equation}\label{eq:defB2}
  \beta_{\rm int}  = 
  z_{\rm int} \left\vert 1-\frac{c_0}{c_1} \right\vert \left( \frac{\sqrt{ D} \sigma_V}{\ell_V}\right)  . 
 \end{equation}

We consider next  the `strong scattering', 'high contrast', or `high-frequency'  regime 
\begin{equation}\label{eq:regime}
|\sigma_V \Delta k|  \gg 1
\end{equation}
for which 
\begin{equation*}
\omega_{\rm o} \left(  \frac{\ell_V\ell_\Psi }
     {c_0 z_{\rm int}} \right) 
 =
\gamma_*
\sigma_V\Delta k  = \omega_{\rm o} (\gamma_*
\sigma_V  |\bcals_0 - \bcals_1|) , 
\qquad
\gamma_*=O(1) .
\end{equation*}
  We then  get via a Laplace's Approximation type argument     
 \begin{align}
    {\rm Var}\left[  \CalI_{\rm mf}(\by)  \right]      & \simeq \sigma_\Psi^2  \left[ \frac{\exp\left( {- \frac{|\by|^2}{ \ell_\Psi^2
     \left( 1 +  (\beta_{\rm int}/\ell_\Psi )^2 \right)} }\right) }{1+  \left( {\beta_{\rm int} }/{\ell_\Psi}\right)^2} 
    \right] . \label{eq:ex1}
 \end{align}
   We give more details of the asymptotic evaluations 
  in the strong scattering  regime \eqref{eq:regime} in Appendix \ref{app:sp}. 
  Expression \eqref{eq:ex1} shows that the support of the incoherent image fluctuations is broadened from the intrinsic source width $\ell_\Psi$  to the effective width
\[
\sqrt{\ell_\Psi^2+\beta_{\rm int}^2}.
\]
Thus $\beta_{\rm int}$  is precisely the scale over which rough-interface scattering spreads the image 
via the associated random wave field modulation or `speckle' like features.  
Several observations follow immediately. First, the broadening grows linearly with the interface depth $z_{\rm int}$. Second, it increases with both the roughness level 
$\sqrt{\CalD}$  and the contrast between the propagation speeds $c_0$ and $c_1$. Third, the parameter $\beta_{\rm int}$ 
is independent of the  source carrier frequency  and depends only on the geometry and statistics of the scattering configuration.

This behavior is a manifestation of the shower-curtain effect at the level of the image {`speckle'}: 
as the wave propagates further from the source to the rough interface, 
the scattering modulation  spreads the image over an increasingly large transverse region.  
The incoherent wave modulation  generated by the scattering
at the rough interface  is not `perfectly time reversed' via the matched  field imaging kernel. 
Consequently, even when the coherent image is strongly suppressed $(|\sigma_V\Delta k|\gg1)$, 
random image modulations  remain of order one inside the broadened region
$
|\by|
\lesssim
\sqrt{\ell_\Psi^2+\beta_{\rm int}^2}.
$
\end{example}

We conclude by noting that the fundamental blurring parameter $ \beta_{\rm int}$ 
will reappear in the optical image, where it controls the deterministic blurring of intensity 
patterns rather than the support of the matched-field speckle.

\section{Optics and the Shower-Curtain Effect in Source Imaging}
\label{sec:opt}

We next consider an optical imaging configuration, illustrated in
Figure~\ref{fig2}, and compare its performance with the matched-field
approach analyzed in the previous section.
The source located at $z=0$ emits a monochromatic wave that propagates
through the rough interface and is subsequently collected by a converging
lens placed in the observation plane $z=z_{\rm obs}$.
The image is formed on a detector plane located at
$z=z_{\rm obs}+z_{\rm i}$, where only the wave intensity is recorded.

As before, the propagation medium is characterized by the piecewise
constant wave speed \eqref{eq:wave_speed}, equal to $c_0$ below the
interface and $c_1$ above it.
Unlike matched-field imaging, which exploits the full complex wave field,
the optical system uses only intensity information. Consequently, phase
distortions induced by the rough interface influence the image only
indirectly through their effect on the spatial distribution of wave
energy.

  \begin{figure}
\begin{center}
\begin{picture}(300,140)
 \linethickness{.35mm}
\put(40,88){\vector(1,0){10}}
\put(40,96){\vector(1,0){10}}
 \thinlines 
\put(10,110){$Source, \Psi$}
   \put(22,85){\vector(0,-1){0}}
 \put(22,85){\vector(0,1){15}}
  \put(24,90){$r_o$}
 \put(5,78){\includegraphics[width=1cm,angle=90]{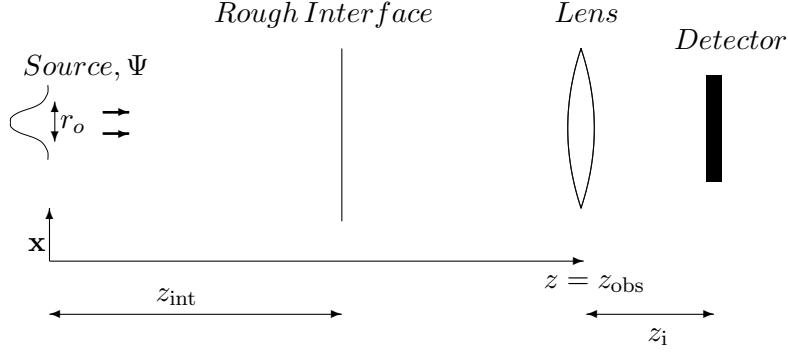}}
 \put(82,130){$Rough \, Interface$}
 \linethickness{2mm}
  \put(270,70){\line(0,1){40}}
 \thinlines 
 \put(255,120){$Detector$}
 \put(210,130){$Lens$}
  \put(20,20){\vector(-1,0){0}}
 \put(20,20){\vector(1,0){110}}
   \put(60,25){$z_{\rm int}$}
  \put(222,20){\vector(-1,0){0}}
  \put(245,10){$z_{\rm i}$}
  \put(222,20){\vector(1,0){48}}
     \qbezier(220,60)(230,90)(220,120)
   \qbezier(220,60)(210,90)(220,120)
   \put(130,55){\line(0,1){65}}
    \put(20,40){\vector(1,0){201}}
    \put(20,40){\vector(0,1){20}}
     \put(12,44){$\bx$}
     \put(206,30){$z=z_{\rm obs}$}
 \end{picture}
\end{center}
\caption{Source imaging in the optical set-up. The detector is a camera or photodetector that records the wave intensity. The source plane $z=0$ and the detector plane $z=z_{\rm int}+z_{\rm i}$ are conjugate, i.e. the focal length $L$ of the lens located in the plane $z=z_{\rm obs}$ satisfies (\ref{eq:L}).  }
\label{fig2}
\end{figure}

It is convenient to introduce the effective propagation distance
\begin{equation*}
    \tilde z_{\rm obs}
    =
    z_{\rm int} \left( \frac{c_0}{c_1} \right)
    +(z_{\rm obs}-z_{\rm int}),
\end{equation*}
which accounts for the different propagation speeds on either side of the
interface. The detector plane is chosen to be conjugate to the source
plane with respect to the lens, so that the focal length  {$L$} satisfies the following lens equation
\begin{align}\label{eq:L}
    \frac{1}{L}
    =
    \frac{1}{\tilde z_{\rm obs}}
    +
    \frac{1}{z_{\rm i}} .
\end{align}

Throughout the analysis we assume for simplicity an ideal full-aperture lens; see
\cite{GS} for the finite-aperture case.
The corresponding transmission function is
\begin{equation*}
    \tau(\bx)
    =
    \exp\left(
      -i\frac{\omega_{\rm o}|\bx|^2}{2c_1L}
    \right).
\end{equation*}

Using the compensated field $U^\e_{\rm obs}$ defined in
\eqref{eq:wobs2b}, together with \eqref{eq:par1}, the optical image formed in the detector plane
$z=z_{\rm obs}+z_{\rm i}$ is then 
\begin{align}\label{eq:v30}
  \CalI_{\rm opt}(\by)
  =
  \left|
  \frac{\omega_{\rm o}}{2\pi i c_1z_{\rm i}}
  \int
  U^\e_{\rm obs}(\bx)\,
  \tau(\bx)\,
  \exp\left(
    i\frac{\omega_{\rm o}|\by-\bx|^2}{2z_{\rm i}c_1}
  \right)
  d\bx
  \right|^2 .
\end{align}

The rough interface modifies the image through scattering of the
transmitted wave so that it becomes partly coherent.  The following result characterizes the resulting mean
optical image.

\begin{proposition}\label{prop3}
The mean image in the optical configuration is :
\begin{itemize}
\item for $\gamma=1/2$
\begin{equation}\label{eq:v3}
 \E\bigl[\CalI_{\rm opt}(\by)\bigr]
  = \left(\frac{\tilde{z}_{\rm obs}}{z_{\rm i}}\right)^{\!2}
    \int\kappa(\Delta k,\bk) \left|\Psi\!\left(-\frac{\tilde{z}_{\rm obs}}{z_{\rm i}}\by
      +\frac{c_0 z_{\rm int}}{\omega_{\rm o}}\bk\right)\right|^2 d\bk,
\end{equation}
\item for $\gamma<1/2$ and $\e\to 0$
\[
  \E\bigl[\CalI_{\rm opt}(\by)\bigr]
  = \left(\frac{\tilde{z}_{\rm obs}}{z_{\rm i}}\right)^{\!2}
    \left|\Psi\!\left(-\frac{\tilde{z}_{\rm obs}}{z_{\rm i}}\by\right)\right|^2,
\]
\item for $\gamma>1/2$ and $\e\to 0$
\[
   \E\bigl[\CalI_{\rm opt}(\by)\bigr]
  = \bigl|\phi_V(\Delta k)\bigr|^2
    \left(\frac{\tilde{z}_{\rm obs}}{z_{\rm i}}\right)^{\!2}
    \left|\Psi\!\left(-\frac{\tilde{z}_{\rm obs}}{z_{\rm i}}\by\right)\right|^2,
\]
\end{itemize}
with $\phi_V$ defined in \eqref{def:charac_V},  $\Delta k$ given in \eqref{eq:dkdef} and   $\kappa$ introduced  in \eqref{eq:kapdef}.
   \end{proposition}
   
 \begin{discussion}

The proposition reveals a marked contrast with the matched-field imaging
results of Section~\ref{sec:match}. Whereas matched-field imaging is
sensitive to random phase distortions, optical imaging depends only on
the redistribution of wave energy.

In the absence of interface fluctuations,
$\kappa(\tau,\bk)=\delta(\bk)$ and $\phi_V(\tau)\equiv1$, yielding
\[
\CalI_{\rm opt}(\by)
=
\left(
\frac{\tilde z_{\rm obs}}{z_{\rm i}}
\right)^2
\left|
\Psi\!\left(
-\frac{\tilde z_{\rm obs}}{z_{\rm i}}\by
\right)
\right|^2.
\]
Thus the optical system reconstructs the source intensity profile exactly,
up to the familiar magnification factor
$\tilde z_{\rm obs}/z_{\rm i}$ and image inversion.

The three asymptotic regimes identified previously for matched-field
imaging reappear here, but with a different interpretation.

\emph{Critical regime ($\gamma=1/2$).}
The interface fluctuations occur on the same scale as the beam width.
The scattered energy is redistributed across propagation directions, and
the optical image becomes a blurred version of the ideal image.
The blurring kernel is precisely $\kappa$, which already appeared in the
variance formula for matched-field imaging.
The width of this blur scales with the effective propagation distance
from the source to the interface and is controlled by the parameter
$\beta_{\rm int}$ introduced in \eqref{def:Bm}.
This regime exhibits the shower-curtain effect most clearly:
the farther the rough interface is from the source, the larger the image
blur. Conversely, when the interface is close to the source, the image
remains relatively sharp despite the presence of strong scattering.

\emph{Large-scale interface fluctuations ($\gamma<1/2$).}
The interface varies on a scale much larger than the beam width.
Across the beam footprint the induced phase perturbation is therefore
approximately constant. Such a phase factor does not affect intensity,
and the optical image is asymptotically identical to that obtained for a
flat interface. In contrast to matched-field imaging, no loss of
coherence is visible at the level of the mean image.

\emph{Small-scale interface fluctuations ($\gamma>1/2$).}
The interface fluctuates on scales much shorter than the beam width.
The rapidly varying phase perturbations average out, and the interface
acts effectively as a homogenized random phase screen.
The image shape is preserved, but its overall contrast is reduced by the
factor $|\phi_V(\Delta k)|^2$.
Consequently, the dominant effect is a loss of coherent energy rather
than geometric blurring.
When $\sigma_V\Delta k\gtrsim1$, the mean image becomes strongly
attenuated and therefore increasingly sensitive to detector noise.

Overall, the optical system exhibits two distinct manifestations of the
rough interface. In the critical regime, the interface causes genuine
image blur whose magnitude grows with distance from the source, giving
rise to the shower-curtain effect. In the rapidly fluctuating regime, the
primary effect is instead a loss of image contrast through incoherent
scattering.

\end{discussion}

\begin{proof}

Using the definition \eqref{eq:v30} of the optical image and expanding the
modulus squared gives
\begin{align*}
\E[\mathcal I_{\rm opt}(\by)]
=
\frac{\omega_{\rm o}^2}{(2\pi c_1 z_{\rm i})^2}
\iint &
\E\!\left[
U^\e_{\rm obs}(\bx)
\overline{U^\e_{\rm obs}(\bx')}
\right]\exp\!\left(
-\frac{i\omega_{\rm o}}{2c_1L}
\bigl(|\bx|^2-|\bx'|^2\bigr)
\right)
\nonumber\\
&\times
\exp\!\left(
\frac{i\omega_{\rm o}}{2c_1z_{\rm i}}
\bigl(|\by-\bx|^2-|\by-\bx'|^2\bigr)
\right)
d\bx\,d\bx'.
\end{align*}
Introducing the midpoint and offset variables
\[
\bs=\frac{\bx+\bx'}2,
\qquad
\bs'=\bx-\bx',
\]
and using
\[
|\bx|^2-|\bx'|^2 = 2\bs\cdot\bs',
\qquad
|\by-\bx|^2-|\by-\bx'|^2
=
-2(\by-\bs)\cdot\bs',
\]
we obtain
\begin{align}
\E[\mathcal I_{\rm opt}(\by)]
&=
\frac{\omega_{\rm o}^2}{(2\pi c_1 z_{\rm i})^2}
\iint
\E\!\left[
U^\e_{\rm obs}\!\left(\bs+\frac{\bs'}2\right)
\overline{
U^\e_{\rm obs}\!\left(\bs-\frac{\bs'}2\right)}
\right]
\nonumber\\
&\qquad\times
\exp\!\left[
-i\frac{\omega_{\rm o}}{c_1L}\bs\cdot\bs'
-i\frac{\omega_{\rm o}}{c_1z_{\rm i}}
(\by-\bs)\cdot\bs'
\right]
d\bs\,d\bs'.
\label{eq:opt_pf2}
\end{align}
Substituting the second-order field moment \eqref{eq:2ndmom} into
\eqref{eq:opt_pf2} yields
\begin{align*}
\E[\mathcal I_{\rm opt}(\by)]
=
\frac{\omega_{\rm o}^2}
     {(2\pi)^8(c_1z_{\rm i})^2}
 \int\dots&\int
e^{-i\frac{\omega_{\rm o}}{c_1L}\bs\cdot\bs'}
e^{-i\frac{\omega_{\rm o}}{c_1z_{\rm i}}e^{i(\bk-\bq)\cdot\bs}
e^{i(\bk+\bq)\cdot\bs'/2}
(\by-\bs)\cdot\bs'}
\nonumber\\
&\times
 {e^{-i\br\cdot(\bq-\bq')}}
\phi_{\Delta V}
\!\left(
\Delta k,
\frac{\br}{\e^{\gamma-1/2}}
\right)
\nonumber\\
&\times
\hat G_1(\bk )
\hat G_0(\bk-(\bq-\bq') )\overline{\hat G_1(\bq )}
\,
\overline{\hat G_0(\bq' )}
\\
&\times 
\hat\Psi(\bk-(\bq-\bq'))
\overline{\hat\Psi(\bq')}
\,
d\bk\,d\bq\,d\bq'\,d\br\,d\bs\,d\bs' . \nonumber
\end{align*}
The integration with respect to $\bs$ gives after using the lens equation \eqref{eq:L}
\[
(2\pi)^2
\delta\!\left(
\bk-\bq
-\frac{\omega_{\rm o}}{c_1\tilde z_{\rm obs}}
\bs'
\right).
\]
Using the identity 
$ \tilde z_{\rm obs} = z_{\rm obs}-z_{\rm int} +  z_{\rm int}  c_0/c_1$
and integrating  with respect to $\bs'$ gives then phases that cancel 
with the propagators  in the right part of the medium and produces propagators
in the left part as 
\begin{align}\label{eq:opt_pf4}
\E[\mathcal I_{\rm opt}(\by)]
=
\frac1{(2\pi)^6}
\left(
\frac{\tilde z_{\rm obs}}{z_{\rm i}}
\right)^2
\iiiint &
e^{i(\bk-\bq)\cdot
\left(
-\frac{\tilde z_{\rm obs}}{z_{\rm i}}\by
\right)}
 {e^{-i\br\cdot(\bq-\bq')}} \phi_{\Delta V}
\!\left(
\Delta k,
\frac{\br}{\e^{\gamma-1/2}}
\right)
\nonumber\\
&\times
\hat G_0(\bk-(\bq-\bq') )
\overline{\hat G_0(\bq' )} \overline{\hat G_0(\bk )}
\hat G_0(\bq ) \\
& \times
\hat\Psi(\bk-(\bq-\bq'))
\overline{\hat\Psi(\bq')}
\,
d\bk\,d\bq\,d\bq'\,d\br .\nonumber
\end{align}
By comparing with \eqref{eq:varE} we see that 
 remaining calculation is
identical to the evaluation of the second moment of the matched-field
image, with the image coordinate
\[
\by
\quad\text{replaced by}\quad
-\frac{\tilde z_{\rm obs}}{z_{\rm i}}\by .
\]
Applying the  asymptotic limits established in
Proposition~\ref{prop2} therefore yields
\eqref{eq:v3}. In particular,

\begin{itemize}
\item for $\gamma=1/2$, the limiting kernel is $\kappa$, leading to the
convolutional blur in \eqref{eq:v3};

\item for $\gamma<1/2$, the kernel
$\phi_{\Delta V}(\Delta k,\br/\e^{\gamma-1/2})$
converges to $1$ and the ideal optical image is recovered;

\item for $\gamma>1/2$, the kernel converges to
$|\phi_V(\Delta k)|^2$, producing the attenuation factor
$|\phi_V(\Delta k)|^2$ multiplying the ideal image.
\end{itemize}

This proves the proposition.

\end{proof}

\subsection{Statistical stability of the optical image }


We now characterize the statistical stability of the optical image,
or equivalently its signal-to-noise ratio.
In contrast to the matched-field functional considered in the previous section,
the optical image is formed from wave intensities and therefore depends quadratically
on the wave field.
Consequently, the variance of the optical image involves fourth-order moments of the
random transmission operator.

To obtain explicit expressions we focus on the strong-scattering regime
\eqref{eq:regime}, together with the critical scaling $\gamma=1/2$,  {and assuming from now on $V$ to be a Gaussian random field}.
As already seen in Example~\ref{ex1}, this is the regime in which the
rough interface produces a nontrivial competition between coherent imaging
and {`speckle'} formation.
The following lemma, for which the proof is deferred to Appendix~\ref{app:lem}, provides the asymptotic fourth-order moment 
of the scattering kernel needed in the
variance calculation.

\begin{proposition}\label{lem1}
For $\gamma=1/2$, we have in the strong-scattering regime \eqref{eq:regime}:
\begin{align}
  \E\Bigl[
    \mathfrak{K}^\e(\tau,\bk,\bk')\,
    \overline{\mathfrak{K}^\e(\tau,\bq,\bq')} & \,
    \overline{\mathfrak{K}^\e(a\tau,\tbk,\tbk')}\,
    \mathfrak{K}^\e(a\tau,\tbq,\tbq')
  \Bigr]
  \nonumber\\
  &\quad=
    \delta\!\bigl((\bk'-\bk)-(\bq'-\bq)\bigr)\,
    \delta\!\bigl((\tbk'-\tbk)-(\tbq'-\tbq)\bigr)
  \nonumber\\
  &\qquad\times
    \delta\!\bigl(a(\bk'-\bk)-(\tbk'-\tbk)\bigr)\,
    \frac{
      \exp\!\left(
      -\dfrac{ |\tbq'-\tbq |^2}
      {2(a\tau)^2\CalD}
      \right)}
    {2\pi\tau^2\CalD},
\label{eq:4K}
\end{align}
where $\CalD$ is defined in Remark~\ref{remark1}.
\end{proposition}

Using Proposition~\ref{lem1}, and proceeding along the same lines as in the proof of Proposition~\ref{prop3}, we obtain the following characterization of the mean image and its covariance.

\begin{proposition}\label{prop4}
In the strong-scattering regime \eqref{eq:regime} and for $\gamma=1/2$,
the mean optical image is
\begin{align}
\E\!\left[\CalI_{\rm opt}(\by)\right]
& =
\left(
\frac{\tilde z_{\rm obs}}{z_{\rm i}}
\right)^2
\int 
 \frac{1}
{2\pi(\Delta k)^2\CalD} \exp\!\left(
-\dfrac{|\bk|^2}
{2(\Delta k)^2\CalD}
\right)
\nonumber\\
&\hspace{4cm}\times
\left|
\Psi\!\left(
-\frac{\tilde z_{\rm obs}}{z_{\rm i}}\by
+
\frac{c_0 z_{\rm int}}{\omega_{\rm o}}\bk
\right)
\right|^2
d\bk
\nonumber\\
&=
\left(
\frac{\tilde z_{\rm obs}}{z_{\rm i}}
\right)^2
\int
\frac{e^{-|\bk|^2/2}}{2\pi}
\left|
\Psi\!\left(
-\frac{\tilde z_{\rm obs}}{z_{\rm i}}\by
+
\beta_{\rm int}\bk
\right)
\right|^2
d\bk ,
\label{eq:v3b}
\end{align}
where $\CalD$ is defined in \eqref{eq:Ddef}
and $\beta_{\rm int}$ in \eqref{def:Bm}. Moreover, for any constant $a$ the covariance of the optical images
at frequencies $\omega_{\rm o}$ and $a\omega_{\rm o}$ is
\begin{align*}
{\rm Cov}
\bigl[
\CalI_{\rm opt}(\by;\omega_{\rm o})&,
\CalI_{\rm opt}(\by';a\omega_{\rm o})
\bigr]
\nonumber\\
&=
\left(
\frac{\tilde z_{\rm obs}}{z_{\rm i}}
\right)^4
\iint
\left|
\Psi\!\left(
-\frac{\tilde z_{\rm obs}}{z_{\rm i}}\by
+\beta_{\rm int}\bk
\right)
\right|^2
\left|
\Psi\!\left(
-\frac{\tilde z_{\rm obs}}{z_{\rm i}}\by'
+\beta_{\rm int}\tbk
\right)
\right|^2
\nonumber\\
&\hspace{2cm}\times
\left(
\delta(\bk-\tbk)\frac{e^{-|\bk|^2/2}}{2\pi}
-
\frac{e^{-|\bk|^2/2}}{2\pi}
\frac{e^{-|\tbk|^2/2}}{2\pi}
\right)
d\bk\,d\tbk .
\end{align*}
\end{proposition}
\begin{discussion}
A striking consequence of Proposition~\ref{prop4} is that,
in the strong-scattering/high-frequency   regime,
both the resolution and the signal-to-noise ratio of the optical image
are governed by a single parameter, namely the blurring parameter
$\beta_{\rm int}$ introduced in Remark~\ref{remark1}.
This is the same parameter that appeared in Example~\ref{ex1} for
matched-field imaging, showing that the influence of the rough interface
on coherent and incoherent imaging is controlled by the same underlying
geometric mechanism. Note that the blurring parameter in this  case with 
one rough interface fluctuating on the scale of the beam width  is independent 
of the carrier frequency $\omega_0$. 

The mean image is obtained by averaging the ideal image intensity
against a Gaussian kernel of width $\beta_{\rm int}$.
Consequently, $\beta_{\rm int}$ directly quantifies the amount of
resolution loss induced by the rough interface.
Since
\[
\beta_{\rm int}
=
z_{\rm int}
\left|1-\frac{c_0}{c_1}\right|
\sqrt{\CalD},
\]
the image degradation increases linearly with the distance between the
source and the interface, increases with the medium contrast,
and increases with the roughness of the interface fluctuations.
This is precisely the shower-curtain effect:
the farther the scattering layer is from the source,
the larger the blur in the reconstructed image.

To interpret the covariance formula, define
\begin{equation*}
h(\by,\bx)
=
\left(
\frac{\tilde z_{\rm obs}}{z_{\rm i}}
\right)^2
\left|
\Psi\!\left(
-\frac{\tilde z_{\rm obs}}{z_{\rm i}}\by
+
\beta_{\rm int}\bx
\right)
\right|^2 .
\end{equation*}
If $\bX$ is a standard Gaussian random vector in $\R^2$, then
\begin{align}
\E\!\left[\CalI_{\rm opt}(\by)\right]
&=
\E\!\left[h(\by,\bX)\right],
\label{eq:v4}
\\
{\rm Cov}
\!\left[
\CalI_{\rm opt}(\by;\omega_{\rm o}),
\CalI_{\rm opt}(\by';a\omega_{\rm o})
\right]
&=
{\rm Cov}
\!\left[
h(\by,\bX),
h(\by',\bX)
\right].
\label{eq:v4b}
\end{align}
Thus the image fluctuations can be viewed as arising from a single
effective Gaussian displacement of the ideal image intensity profile.
When $h$ varies slowly with respect to its second argument,
equivalently when the image features are broad compared with
$\beta_{\rm int}$,
the covariance is small and the optical image is statistically stable.
Conversely, when $\beta_{\rm int}$ is comparable to the characteristic
image scale, both image blur and image fluctuations become significant.

An important consequence of Proposition \ref{prop4} is that, in the 
critical regime $\gamma=1/2$ with 
strong-scattering  \eqref{eq:regime}, averaging over frequency or image position does not significantly improve
the signal-to-noise ratio. Indeed, the leading-order covariance is independent of the frequency
ratio $a$, so that optical images obtained at different frequencies remain strongly correlated.
Moreover, the covariance decays with the spatial offset $|\by-\by'|$ only through the variation
of the source intensity profile itself. Hence the correlation length of the image fluctuations is
of the same order as the image support.
Physically, the rough interface acts as a single random phase
screen that imprints a common random distortion on the entire image. The resulting fluctuations
therefore remain coherent across both frequency and image position. Consequently, neither
multifrequency averaging nor local spatial smoothing substantially reduces the image variance.
This is another manifestation of the shower-curtain effect: the dominant uncertainty is carried
by large-scale image distortions rather than by rapidly varying speckle fluctuations.
In fact in view of \eqref{eq:v4b} one `effective'  Gaussian random variable $\bX$ 
drives the fluctuations at all frequencies and image locations. Thus the randomness is effectively low-dimensional (a common random image displacement/blurring variable) rather than a high-dimensional speckle field. 
This is precisely why averaging is ineffective: different measurements do not provide statistically independent realizations of the fluctuations. That interpretation fits very naturally with the physical picture of a shower curtain, where moving one's eye slightly or changing color (frequency) does not remove the apparent blur via averaging.
 Indeed in the critical regime the features of the curtain
is on the scale of the source support. 

Finally, it is worth contrasting the critical regime $\gamma=1/2$
with the two other scaling regimes.
For $\gamma<1/2$, the interface fluctuations vary on scales larger than
the beam width and act primarily through an almost uniform random phase.
Since optical imaging depends on intensity rather than phase,
the image remains statistically stable to leading order.
For $\gamma>1/2$, the interface fluctuations vary on scales much smaller
than the beam width and their effect self-averages across the optical
aperture.
The image fluctuations are therefore again weak, although the coherent
image may be attenuated by the factor $|\phi_V(\Delta k)|^2$.
Hence the critical regime $\gamma=1/2$ is the unique scaling in which
nontrivial image blur and nontrivial image fluctuations persist
simultaneously in the strong-scattering/high-frequency limit.

\end{discussion}
 
\begin{proof}[Outline of the proof Proposition \ref{prop4}]
The derivation follows the same strategy as in the proof of Proposition~\ref{prop3}. The essential difference is that the optical image is quadratic in the wave field. Consequently, the second moment of the optical image involves a fourth-order moment.  In the strong-scattering regime \eqref{eq:regime}, the fourth-order moment 
of the scattering operator is characterized by Lemma~\ref{lem1}.

Starting from the definition \eqref{eq:v30}, we write
\begin{align*}
\E\big[
{\mathcal I}{\rm opt}(\by;\omega_{\rm o})
{\mathcal I}{\rm opt}&(\by';a\omega_{\rm o})
\big]
\\
=
\frac{a^2\omega_{\rm o}^4}
{(2\pi c_1 z_{\rm i})^4}
 \iiiint &
\E\Big[
U^\e_{\rm obs}(\bx;\omega_{\rm o})
\overline{U^\e_{\rm obs}(\bx';\omega_{\rm o})}
U^\e_{\rm obs}(\tilde{\bx};a\omega_{\rm o})
\overline{U^\e_{\rm obs}(\tilde{\bx}';a\omega_{\rm o})}
\Big]
\\
&\times
\exp\left(
-\frac{i\omega_{\rm o}}{2c_1L}
\big(|\bx|^2-|\bx'|^2\big)
\right)
\exp\left(
\frac{i\omega_{\rm o}}{2c_1z_{\rm i}}
\big(|\by-\bx|^2-|\by-\bx'|^2\big)
\right)
\\
&\times
\exp\left(
-\frac{ia\omega_{\rm o}}{2c_1L}
\big(|\tilde{\bx}|^2-|\tilde{\bx}'|^2\big)
\right)
\exp\left(
\frac{ia\omega_{\rm o}}{2c_1z_{\rm i}}
\big(|\by'-\tilde{\bx}|^2-|\by'-\tilde{\bx}'|^2\big)
\right) 
 \\
&\times
d\bx d\bx' d\tilde{\bx} d\tilde{\bx}' .
\end{align*}
Then we use  \eqref{eq:wobs2b},  Lemma \ref{lem1}, the lens equation \eqref{eq:L}   
and proceed  as in the proof of Proposition \ref{prop3}
 to  find the generalization of \eqref{eq:opt_pf4}. 
 \begin{align*}
 \E\big[{\mathcal I}_{\rm opt}(\by; \omega_{\rm o}) {\mathcal I}_{\rm opt}&(\by'; a\omega_{\rm o}) \big]\\
  =
\frac{1}{(2\pi)^8}   \left( \frac{\tilde{z}_{\rm obs} }{z_{\rm i}}\right)^4  &
\int\dots \int       e^{i(\bk-\bq) \cdot \left( - \frac{\tilde{z}_{\rm obs} }{z_{\rm i}}   \by \right)}
e^{i(\tilde\bk-\tilde\bq) \cdot \left( - \frac{\tilde{z}_{\rm obs} }{z_{\rm i}}   \by' \right)}
             \nonumber      \\ & \times 
  \hat{G}_0(\omega_{\rm o},\bk-(\bq-\bq') )
 \,  \overline{\hat{G}_0(\omega_{\rm o},\bq' )} \, 
    \overline{\hat{G}_0(\omega_{\rm o},\bk  ) }
 \,   {\hat{G}_0(\omega_{\rm o},\bq )}
            \nonumber      \\ &\times 
  \hat{G}_0(a\omega_{\rm o},\tilde\bk-(\tilde\bq-\tilde\bq') )
 \,  \overline{\hat{G}_0(a\omega_{\rm o},\tilde\bq' )} \, 
    \overline{\hat{G}_0(a\omega_{\rm o},\tilde\bk  ) }
 \,   {\hat{G}_0(a\omega_{\rm o},\tilde\bq )} 
\\ 
& \times   \delta(a(\bq'-\bq) - (\tbq'-\tbq))  \frac{1
}
{
2\pi(\Delta k)^2\CalD
} \exp\left(
-\dfrac{
|\tilde{\bq}'-\tilde{\bq}|^2
}{
2(a\Delta k)^2\CalD
}
\right)
\\
&\times
\hat\Psi\big(\bk-(\bq-\bq')\big)
\overline{\hat\Psi(\bq')}
\hat\Psi\big(\tilde{\bk}-(\tilde{\bq}-\tilde{\bq}')\big)
\overline{\hat\Psi(\tilde{\bq}')} \, 
 d\bk d\bq d\bq'
 d\tilde{\bk} d\tilde{\bq} d\tilde{\bq}' .
\end{align*}
 Then we can conclude in a similar fashion as  in the proof of Proposition \ref{prop3},
with a reduction analogous to the one of deriving   \eqref{eq:v3}  from  \eqref{eq:opt_pf4}.
This then gives the second moment of the optical imaging function.
The expression for the mean  image in \eqref{eq:v3b}   follows from the expression \eqref{eq:v3} 
upon substituting with the asymptotic form of the scattering kernel  $\kappa$ 
on the strong-scattering regime. 
Finally, subtracting the product of the means yields the covariance formula stated
 in the proposition. 
\end{proof}

\subsection{Illustrations}

The results of the previous subsection show, in the critical regime
$\gamma=1/2$, with Gaussian interface fluctuations and source profile \eqref{eq:defSV}, that under the strong-scattering condition
\eqref{eq:regime}, the effect of the rough interface on the optical image
is governed by a single dimensionless quantity,
\[
\frac{\beta_{\rm int}}{\ell_\Psi},
\]
which compares the characteristic blur induced by the interface with the
intrinsic width of the source.  The parameter
$\beta_{\rm int}$ therefore plays the role of an effective
``shower-curtain'' length scale for source imaging.

\begin{example31B}
We revisit the configuration of Example~\ref{ex1} corresponding to
$\gamma=1/2$.  In the high-frequency regime \eqref{eq:regime},
the representation \eqref{eq:v4} gives
\begin{align*}
\E\!\left[\CalI_{\rm opt}(\by)\right]
&\sim
\sigma_\Psi^2
\left(
\frac{\tilde z_{\rm obs}}{z_{\rm i}}
\right)^2
\frac{
\exp\!\left(
-\dfrac{
\left(
\frac{\tilde z_{\rm obs}}{z_{\rm i}}
\right)^2
|\by|^2
}{
 \ell_\Psi^2\bigl(1+(\beta_{\rm int}/\ell_\Psi)^2\bigr)
}
\right)
}{
1+(\beta_{\rm int}/\ell_\Psi)^2
}.
\end{align*}
The image therefore remains Gaussian but with an enlarged width
\[
\ell_{\rm eff}
=
\sqrt{
\ell_\Psi^2+ \beta^2_{\rm int}
}.
\]
The rough interface thus acts as an additional random imaging system
whose blur radius is precisely $\beta_{\rm int}$.

Let
\[
\eta:=\frac{\beta_{\rm int}}{\ell_\Psi}
\qquad \text{and} \qquad 
\rho(\by)
:=
\left(
\frac{\tilde z_{\rm obs}}{z_{\rm i}}
\right)^2
\frac{|\by|^2}{\ell_\Psi^2},
\]
then \eqref{eq:v4}--\eqref{eq:v4b} yield the signal-to-noise ratio
\begin{align}
{\rm SNR}(\by)
&=
\frac{\E^2[\CalI_{\rm opt}(\by)]}
{{\rm Var}[\CalI_{\rm opt}(\by)]}
\nonumber\\
&=
\left(
\frac{(1+2 \eta^2)^2}{1+4\eta^2}
\exp\!\left(
\frac{
4 \eta^2\,\rho(\by)
}{
(1+2\eta^2)(1+4\eta^2)
}
\right)
-1
\right)^{-1}.
\label{eq:SNR}
\end{align}
In particular,
\[
\lim_{\eta\downarrow0}
{\rm SNR}(\by)=\infty,
\]
which corresponds to the ideal flat-interface case.
As $\eta$ increases, the rough-interface blur becomes comparable to the
source width and the image fluctuations become increasingly important.
Moreover, the signal-to-noise ratio decreases away from the image center,
reflecting the fact that the coherent image decays more rapidly than the
incoherent {`speckle'} contribution.

Figure~\ref{figSNR} shows the signal-to-noise ratio
\eqref{eq:SNR} for various values of the parameters
$\eta$ and $\rho$.

\begin{figure}
\centering
\resizebox*{6cm}{!}{\includegraphics{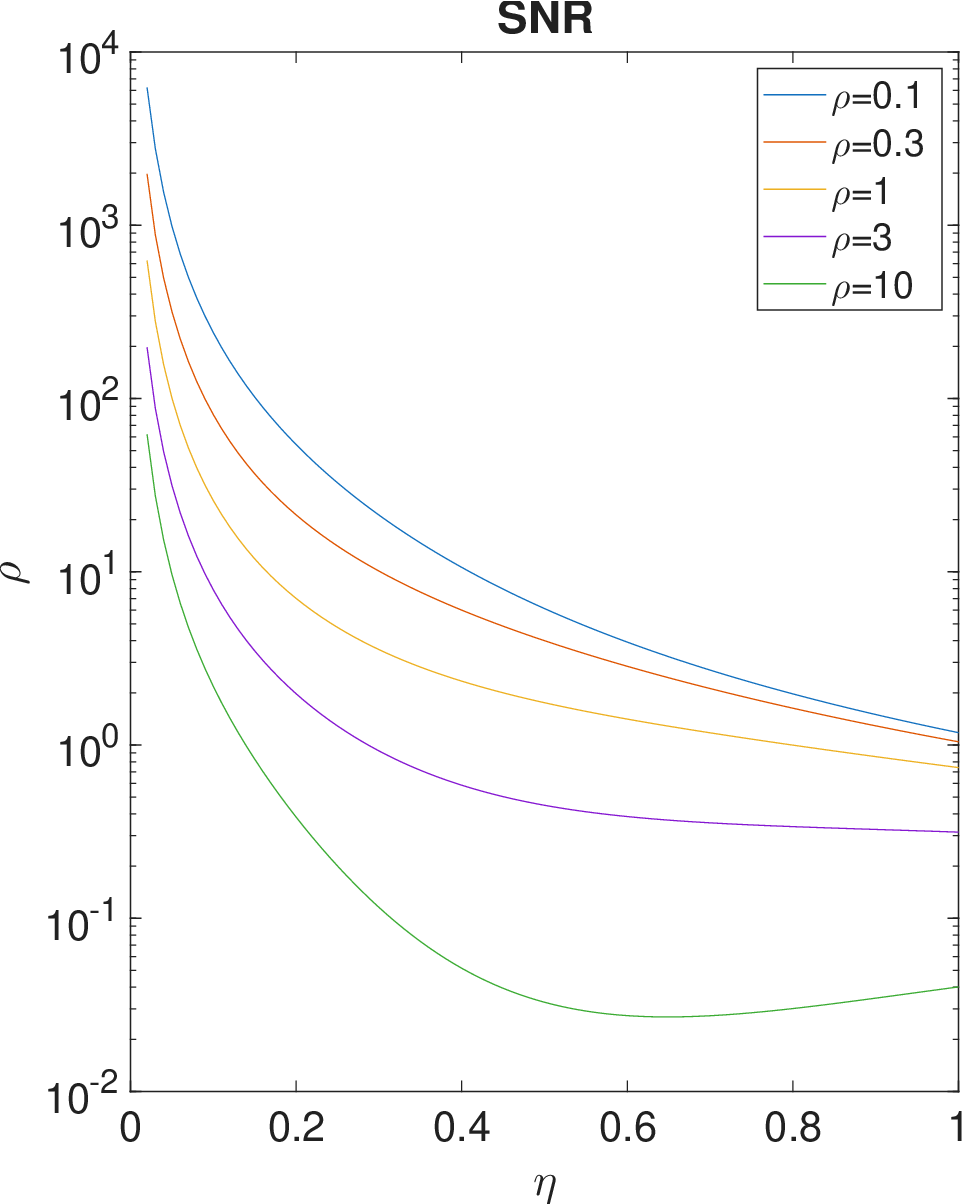}}
\caption{
Signal-to-noise ratio \eqref{eq:SNR} as a function of the relative
blurring parameter $\eta=\beta_{\rm int}/\ell_\Psi$
and the normalized image offset $\rho$.
Increasing $\eta$ corresponds to stronger shower-curtain blurring and
reduces the image stability.
}
\label{figSNR}
\end{figure}

\end{example31B}

\begin{example11B}

We finally return to the experimental configuration shown in
Figure~\ref{fig1}.  From the measured images we estimate the effective
blur width as a function of the distance between the
1951 USAF resolution chart and the rough interface.

The theory predicts, in {the strong-scattering regime, or equivalently hight-contrast or high-frequency regime, that the blurring is again  characterized by the parameter} 
\eqref{eq:regime},
\[
\beta_{\rm int}
=
z_{\rm int}
\left|
1-\frac{c_0}{c_1}
\right|
\sqrt{\CalD},
\]
so that the blur width should increase linearly with the separation
between the object and the rough interface,  {but also as the contrast or the roughness of the interface increase}.

The experimentally (roughly) estimated blur widths are shown in
Figure~\ref{figbeta} (red circles) together with the linear model
\[
\beta_{\rm int}
=
\alpha z_{\rm int},
\]
with $\alpha=0.1$ (blue line).

The observed trend is consistent with the theoretical prediction.
Most importantly, the figure illustrates the essential geometric nature
of the shower-curtain effect: the farther the object is from the rough
interface, the larger the random angular deviations accumulated during
back-propagation and hence the broader the resulting image blur.
Conversely, when the object is placed close to the interface, the random
phase distortions have insufficient propagation distance to develop into
substantial image degradation.

\begin{figure}
\centering
\resizebox*{6cm}{!}{\includegraphics{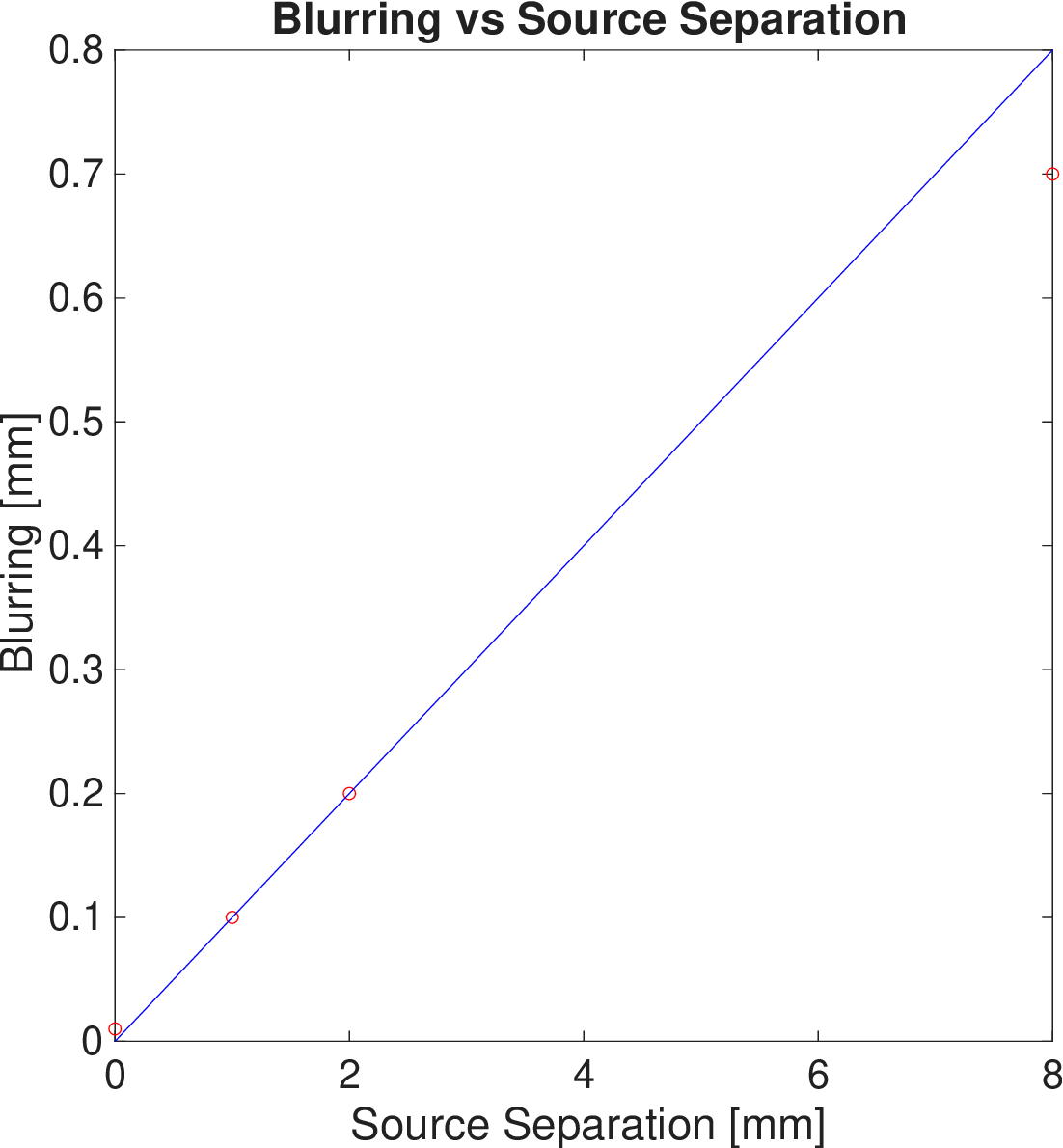}}
\caption{
Estimated blur width as a function of the distance between the
1951 USAF resolution chart and the rough interface.
The approximately linear growth agrees with the theoretical prediction
$\beta_{\rm int}\propto z_{\rm int}$ and provides direct evidence of the
shower-curtain effect.
}
\label{figbeta}
\end{figure}

\end{example11B}

\section{Conclusions}\label{sec:concl}

We have presented a mathematical theory for the classical shower-curtain effect in wave and optical imaging through a rough interface. The analysis provides an explicit characterization of how image quality depends not only on the statistical properties of the interface, but also on its geometric location relative to the source and detector. In particular, the theory quantifies how image resolution and image stability degrade as a consequence of random scattering induced by the interface.

A central outcome of the analysis is the identification of a fundamental blurring parameter,
\[
\beta_{\rm int}= z_{\rm int}
\left|1-\frac{c_0}{c_1}\right|
\sqrt{\CalD},
\]
with $\CalD$ defined by \eqref{eq:Ddef}, which combines the effects of propagation distance, medium contrast, and interface roughness into a single quantity. This parameter governs both the spatial broadening of the image and the magnitude of the image fluctuations. The matched-field and optical imaging analyses lead to the same central conclusion: the severity of the shower-curtain effect is not determined solely by the roughness of the interface itself. Rather, it is controlled by the combined effect of the propagation distance, contrast, and interface roughness, encapsulated by 
$\beta_{\rm int}$.

The results provide a quantitative explanation of a familiar physical observation. When a rough interface is located close to the source, the resulting image degradation is relatively weak. As the distance between the source and the interface increases, the random scattering generated by the interface 
cannot be appropriately compensate for, 
leading to a broader image and reduced image stability. In this sense, the shower-curtain effect is fundamentally a propagation phenomenon driven by random scattering. Consistent with this interpretation, the effect disappears entirely for a flat interface.

An important feature of the  {strong-scattering/high-contrast/high-frequency regime} considered here is that both the image resolution and the signal-to-noise ratio become asymptotically independent of the probing frequency. In the strong-scattering regime, the image statistics acquire a universal form and depend primarily on the single parameter $\beta_{\rm int}$. This universality makes it possible to obtain explicit formulas for image resolution, image fluctuations, and signal-to-noise ratio.

The comparison between matched-field and optical imaging is also noteworthy. Matched-field imaging exploits the full complex wave field, whereas optical imaging relies only on intensity measurements.
Optical imaging is found to be significantly less sensitive to random interface fluctuations. This robustness arises because large-scale random phase distortions have only an indirect effect on measured intensities. The analysis therefore provides a quantitative explanation for the  stability of optical imaging in the presence of random scattering.

The present work  has several limitations. The analysis is restricted to a single rough interface and to the asymptotic scaling regimes considered here. More general geometries, multiple scattering layers, volumetric random media, and interfaces with lower regularity and larger amplitude remain outside the scope of the current theory. In addition, the source model is coherent and time-harmonic, whereas    imaging systems may employ partially coherent illumination, broadband sources, or active probing strategies. Extending the analysis to such settings remains an important challenge.

The framework developed here nevertheless suggests a number of promising directions for future work. One natural extension is the study of imaging through combinations of rough interfaces and random media, where both volume scattering and interface scattering contribute to image degradation. Another interesting direction is matched-field imaging based on speckle correlations, which may retain useful information even in strongly scattering environments where coherent imaging becomes ineffective. It would also be valuable to investigate partially coherent and broadband illumination with oblique incidence,  
as well as active imaging configurations.

Beyond the specific problem considered here, the theory is relevant to a broad range of imaging applications involving propagation through complex structures, including imaging through biological tissue, walls, atmospheric layers, and other random interfaces. The results are also pertinent to emerging imaging concepts in which inexpensive random structures are intentionally introduced in front of a detector to create an effective aperture larger than the physical sensing device. In such settings, understanding the interplay between random scattering and configuration geometry on the one hand and image resolution, and statistical stability on the other 
is essential, and the present work provides a quantitative framework for doing so.

\section*{Ackowledgements}
\begin{itemize}
\item
 C. Gomez  received support from the french government under the France 2030
investment plan, as part of the Initiative d'Excellence d'Aix-Marseille Universit\'e -- A${}^\ast$MIDEX, through REALISE project, number AMX-22-RE-AB-057.
\\
 K. S{\o}lna  was supporter by the Air Force award FA9550261B189  issued by the Air Force Office of Scientific Research.
 \item
 We thank  Pei, Shan,  \& Xie for permission to use their Figure S1 in \citep{xie2}.
 \item
 Disclosure: ChatGPT was used to rephrase and polish text and to streamline some calculations. 
\end{itemize}

\appendix

\section{Proof of the Jump Conditions and Continuity Relations}

This appendix establishes the jump conditions induced by the source term at
the plane $z=0$ and the continuity relations satisfied by the wave field across
the randomly perturbed interface.

\subsection{Jump conditions at the source plane $z=0$}
\label{app:jump}

Consider the transverse Fourier transform introduced in
\eqref{def:FourierNew},
\begin{equation*}
\hp(\bk,z)
=
\int_{\mathbb R^2}
\check u(\bx,z)\,
e^{-i\bk\cdot\bx}\,d\bx .
\end{equation*}
Then $\hp$ satisfies to the left of the interface the one-dimensional Helmholtz equation 
\begin{equation}
\label{eq:Helmholtz2}
\partial_z^2 \hp(\bk,z)
+
\left(
\frac{\omega_{\rm o}\bcals_0^\e(\bk)}{\e}
\right)^2
\hp(\bk,z)
=
\hat\Psi(\bk)\,\delta'(z),
\end{equation}
with $\bcals_0^\e$ defined in \eqref{def_lambda}.

Since the source is localized at $z=0$, the jump condition is determined
entirely by the local behavior of the medium near the source plane.
Consequently, we may ignore the rough interface in the
derivation below without affecting the resulting jump conditions.

Away from $z=0$, the right-hand side of \eqref{eq:Helmholtz2} vanishes, and
the solution is a superposition of outgoing plane waves. Imposing radiation
conditions as $z\to\pm\infty$, we write
\begin{equation}
\label{eq:dec_hp}
\hp(\bk,z)
=
\Tbu(\bk)
e^{-i\omega_{\rm o}\lwu(\bk)z/\e}
\,\mathbf 1_{(-\infty,0)}(z)
+
\Tau(\bk)
e^{i\omega_{\rm o}\lwu(\bk)z/\e}
\,\mathbf 1_{(0,\infty)}(z) . 
\end{equation}
Substituting \eqref{eq:dec_hp} into \eqref{eq:Helmholtz2} and differentiating
in the sense of distributions, the homogeneous Helmholtz equation is satisfied 
away from $z=0$, leaving only the singular contributions
\begin{align*}
\hat\Psi(\bk)\,\delta'(z)
&=
\left(
\frac{2i\omega_{\rm o}\lwu(\bk)}{\e}
\right)
\bigl(\Tau(\bk)+\Tbu(\bk)\bigr)\delta(z)
\\
&\qquad
+
\bigl(\Tau(\bk)-\Tbu(\bk)\bigr)\delta'(z).
\end{align*}
The coefficients  of the distributions  $\delta$ and $\delta'$ 
must agree, and we therefore obtain
\begin{align*}
\Tau(\bk)+\Tbu(\bk)
&=0,
\nonumber\\
\Tau(\bk)-\Tbu(\bk)
&=\hat\Psi(\bk).
\end{align*}
Solving these equations gives
\begin{equation*}
\Tau(\bk)
=
\frac{1}{2}\hat\Psi(\bk),
\qquad
\Tbu(\bk)
=
-\frac{1}{2}\hat\Psi(\bk).
\end{equation*}
In particular,
\begin{equation}
\hp(\bk,0^+)
=
\Tau(\bk)
=
\frac{1}{2}\hat\Psi(\bk), \label{eq:icA}
\end{equation}
which is precisely the initial condition \eqref{eq:IC}.

We note that the source term is proportional to $\delta'(z)$ rather than
$\delta(z)$. Consequently, the field itself exhibits a jump across the source
plane, while the radiation condition determines the amplitudes of the outgoing
waves on either side.

\subsection{Continuity relations across the randomly perturbed interface}
\label{app:scatt}

We next justify the continuity relations
\eqref{eq:continuity_relation} across the rough interface
\[
z=z_{\rm int}(\bx)=z_{\rm int} + \e V\left(\frac{\bx}{\e^{\gamma-\frac12}}\right),
\]
and refer to \cite{gomezsolna26} for a complete discussion. 
Away from the source plane, the wave field satisfies the homogeneous
Helmholtz equation
\begin{equation}
\label{eq:wave3}
\left(
\partial_{zz}^2+\e^{-1}\Delta_\perp
\right)\check u(\bx,z)
+
\frac{\omega_{\rm o}^2}{\e^2 c^2(\bx,z)}
\check u(\bx,z)
=
0 .
\end{equation}
The coefficient $c(\bx,z)$ is piecewise constant  and exhibits a jump only across
the interface. Consequently, any singular
distributional contribution in \eqref{eq:wave3} can only arise from
discontinuities of $\check u$ or of its first derivatives across the
interface.

Let
\[
\Gamma^\e
=
\left\{
(\bx,z):
z=z_{\rm int}(\bx)
\right\}
\]
denotes the interface  {and let $\bn$ be its unit normal vector}. Suppose that
$\check u$ possesses jumps across $\Gamma^\e$. Writing \eqref{eq:wave3} in
distributional form, a jump of $\check u$ would generate a
${\delta'_{\eta}}$ contribution through the second derivatives, while a
jump of the normal derivative ${\partial_{\eta}\check u}$ would generate a
${\delta_{\eta}}$ contribution.

Since the coefficient term
\[
\frac{\omega_{\rm o}^2}{\e^2 c^2(\bx,z)}\check u
\]
contains no derivatives, it cannot cancel such singular terms,
 all singular contributions must therefore be absent. It follows
that
\[
[\check u]_{\Gamma^\e}=0,
\qquad
[ {\partial_{\eta}}\check u]_{\Gamma^\e}=0,
\]
where $[\cdot]_{\Gamma^\e}$ denotes the jump across the interface.
Hence the wave field and its normal derivative are continuous across the
rough interface,
\begin{equation*}
\check u_+=\check u_-,
\qquad
 {\partial_{\eta}}\check u_+
=
 {\partial_{\eta}}\check u_- ,
\end{equation*}
and these relations together  lead to  the continuity relations stated in
\eqref{eq:continuity_relation}.

 \section{Identification of the Effective Interface Scattering Operator}
\label{app:iscatt}

In this appendix we derive the representation \eqref{eq:wobs} for the
transmitted wave field and identify the effective scattering operator
associated with the rough interface.

We introduce wave amplitudes whose phases are centered relative to the mean
interface location $z=z_{\rm int}$. Specifically,
\begin{align}  
\check u_{\rm inc}(\bx,z)
&=
\frac{1}{(2\pi)^2}
\int
e^{i\bx\cdot\bk}
\hau(\bk)
e^{i\omega_{\rm o}\bcals_0^\e(\bk)(z-z_{\rm int})/\e}
 d\bk ,  \label{eq:cu}
 \\ \nonumber 
\check u_{\rm ref}(\bx,z)
&=
\frac{1}{(2\pi)^2}
\int
e^{i\bx\cdot\bk}
\hbr(\bk)
e^{-i\omega_{\rm o}\bcals_0^\e(\bk)(z-z_{\rm int})/\e}
 d\bk ,
 \\ 
\check u_{\rm tr}(\bx,z)
&=
\frac{1}{(2\pi)^2}
\int
e^{i\bx\cdot\bk}
\hatr(\bk)
e^{i\omega_{\rm o}\bcals_1^\e(\bk)(z-z_{\rm int})/\e}
 d\bk .  \label{eq:cu2}
\end{align}
see Figure \ref{figAc}. 
The total wave field is then expressed for $z>0$ by
\begin{equation*}
    \hat u(\bx,z)  = 
\bigl(
\check u_{\rm inc}(\bx,z)
+
\check u_{\rm ref}(\bx,z)
\bigr)
 \mathbf 1_{(0,z_{\rm int}(\bx))}(z)
+
\check u_{\rm tr}(\bx,z)
\mathbf 1_{(z_{\rm int}(\bx),\infty)}(z).
\end{equation*}
Note that in view of \eqref{eq:icA} and \eqref{eq:cu} we have
\begin{align}\label{eq:ao}
    \hau(\bk) =  \Tau(\bk) e^{i\omega_{\rm o}\lwu(\bk) z_{\rm int}/\e}  = 
\frac12
\hat\Psi(\bk)
e^{i\omega_{\rm o}\bcals^\e_0(\bk)z_{\rm int}/\e}.
\end{align}

It remains to determine the transmitted and reflected amplitudes
$\hatr$  and $\hbr$.
 The continuity conditions \eqref{eq:continuity_relation} imply that the
incident, reflected and transmitted fields must match at
\begin{equation}\label{def_zint}
z=z_{\rm int}(\bx)=z_{\rm int}
+ \e V\left(\frac{\bx}{\e^{\gamma-\frac12}}\right) :=
z_{\rm int} + 
\e \Delta z(\bx) ,
 \end{equation}
 so that we have 
 \begin{equation*} 
\check{u}(\bx, z )_{|z=z_{\rm int}(\bx)^+} = \check{u}(\bx, z)_{|z=z_{\rm int}(\bx)^-}
\end{equation*}
and
\begin{equation*}
\partial_z 
\check{u}(\bx, z)_{|z=z_{\rm int}(\bx)^+} = \partial_z \check{u}(\bx, z)_{|z=z_{\rm int}(\bx)^-}  .
\end{equation*} 
 To leading order in the paraxial scaling we then get
\begin{align*}
   e^{ i\omega_{\rm o}s_0 \Delta  z(\bx) }  \ \check u_{\rm inc} (\bx,z_{\rm int} ) 
 + e^{-  i\omega_{\rm o} s_0 \Delta  z(\bx) } 
  \ \check u_{\rm ref} (\bx,z_{\rm int} ) & = 
  e^{ i\omega_{\rm o}s_1 \Delta  z(\bx)  }  \  
 \check u_{\rm tr} (\bx,z_{\rm int} )  ,  \\
   s_0  \left(  e^{ i\omega_{\rm o}s_0 \Delta  z(\bx) }  \ \check u_{\rm inc} (\bx,z_{\rm int} ) 
 -  e^{-  i\omega_{\rm o} s_0 \Delta  z(\bx) } 
  \ \check u_{\rm ref} (\bx,z_{\rm int} )  \right)    &    = 
  s_1  e^{ i\omega_{\rm o}s_1 \Delta  z(\bx)  }  \  
 \check u_{\rm tr} (\bx,z_{\rm int} )  . 
 \end{align*}
This then gives to leading order  the relations
\begin{align}\label{eq:tran}
    \check u_{\rm tr} (\bx,z_{\rm int} )  &  \simeq  \CalT \ 
     e^{  i\omega_{\rm o}  (s_0-s_1)  \Delta  z(\bx) }  \   \check u_{\rm inc} (\bx,z_{\rm int} ), \\
         \check u_{\rm ref} (\bx,z_{\rm int} )  &  \simeq  {\mathcal R} \ 
     e^{ 2  i\omega_{\rm o}  s_0  \Delta  z(\bx) }  \   \check u_{\rm inc} (\bx,z_{\rm int} ), \nonumber 
\end{align}
 with the effective background   interface transmission   and refection coefficients defined by 
 \begin{equation*}
\CalT := \frac{2  {\bcals_0  }}{\bcals_0+\bcals_1}= \frac{2  { c_1 }}{c_0+c_1} 
\qquad\text{and}\qquad {\mathcal R} =  \frac{  c_1  -  c_0}{c_0+c_1} . 
\end{equation*}
The key point is that the rough interface enters only through the random
phase factor
\[
e^{i\omega_{\rm o}(\bcals_0-\bcals_1)\Delta z(\bx)},
\]
which modulates the transmitted field.

Using \eqref{eq:cu2}, \eqref{eq:tran} and \eqref{eq:cu}   we find
\begin{align*}
\hatr(\bk)
&=
\int
e^{-i\bx\cdot\bk}
\check u_{\rm tr}(\bx,z_{\rm int})
d\bx
\\
&\simeq
\CalT
\int
e^{-i\bx\cdot\bk}
e^{i\omega_{\rm o}(\bcals_0-\bcals_1)\Delta z(\bx)}
\check u_{\rm inc}(\bx,z_{\rm int})
d\bx
\\
& \simeq
\frac{\CalT}{(2\pi)^2}
\iint
e^{-i\bx\cdot(\bk-\bq)}
e^{i\omega_{\rm o}(\bcals_0-\bcals_1)\Delta z(\bx)}
\hau(\bq)
 d\bq d\bx .
\end{align*}
Note next that the expansion in  \eqref{eq:exp_lambda}  gives  
\begin{align*}
    e^{i\omega_{\rm o}\bcals_j^\e(\bk) z /\e}   \simeq 
       e^{i  \omega_{\rm o} \bcals_j  z  /\e}\  e^{- i  \frac{  c_j   | \bk |_2^2 z}{2 \omega_{\rm o}}}
       =
       e^{i  \omega_{\rm o} \bcals_j z  /\e} \  \hat{G}_j(\omega_{\rm o},\bk,z) , 
\end{align*}
for $\hat{G}_j$ defined  in \eqref{eq:hG}.
Using again \eqref{eq:cu2}, \eqref{eq:ao}, and \eqref{def_zint}, and recalling  the expression
\eqref{eq:K} for the scattering operator we then find
  for the transmitted wave field the approximation at $z> z_{\rm int}(\bx)$ 
\begin{align*}
\check u_{\rm tr}(\bx,z)
&\simeq
\frac{\CalT}{2(2\pi)^2}
e^{i\omega_{\rm o}(
\bcals_0 z_{\rm int}
 +
\bcals_1(z-z_{\rm int})
)/\e} \\
& \times \iint
e^{i\bx\cdot\bk}
\hat G_1(\omega_{\rm o},\bk,z-z_{\rm int})
\mathfrak K^\e
\bigl(
\omega_{\rm o}(\bcals_0-\bcals_1),
\bk,\bq
\bigr)
\hat G_0(\omega_{\rm o},\bq,z_{\rm int})
\hat\Psi(\bq)
 d\bq d\bk .
\end{align*}
Evaluating this expression in the observation plane yields the field
representation \eqref{eq:wobs}.
 \begin{figure}
\begin{center}
\begin{picture}(300,140)
 \linethickness{.35mm}
 \put(50,96){\vector(1,0){30}}
  \put(53,85){$\hau$}
  \put(115,81){\vector(-1,0){30}}
  \put(101,65){$\hbr$}
  \put(125,96){\vector(1,0){30}}
  \put(128,85){$\hatr$}
 \thinlines 
\put(10,110){$Source, \Psi$}
   \put(22,85){\vector(0,-1){0}}
 \put(22,85){\vector(0,1){15}}
  \put(24,90){$r_o$}
    \put(5,78){\includegraphics[width=1cm,angle=90]{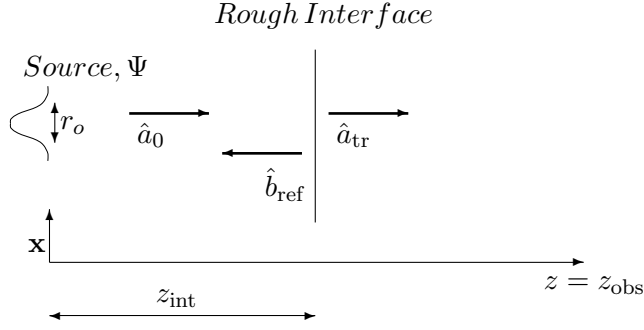}} 
  \put(82,130){$Rough \,  Interface$}
 \linethickness{2mm}
 \thinlines 
  \put(20,20){\vector(-1,0){0}}
 \put(20,20){\vector(1,0){100}}
  \put(60,25){$z_{\rm int}$}
    \put(120,55){\line(0,1){65}}
     \put(20,40){\vector(1,0){201}}
    \put(20,40){\vector(0,1){20}}
     \put(12,44){$\bx$}
     \put(206,30){$z=z_{\rm obs}$}
 \end{picture}
\end{center}
\caption{ The kernels of the impinging,  the transmitted and the reflected wave fields.    }
\label{figAc}
\end{figure}
 
\section{A Laplace Approximation in the Strong-Scattering Regime}
\label{app:sp}

The purpose of this appendix is to justify the asymptotic evaluation
used in Example~\ref{ex1}. The argument is in fact more general and
shows, in the strong-scattering regime, that the matched-field image
variance converges to a Gaussian smoothing of the source intensity
profile. The width of this smoothing kernel is precisely the
blurring parameter $\beta_{\rm int}$.

\begin{proposition}
Assume Assumption~\ref{ass1},  the parameterization  \eqref{eq:defSV} and let
\[
H_{\by}(\br)
=
\frac{1}{(2\pi)^2}
\iint 
e^{i\bk\cdot\br}
\left|
\Psi(\by+\bk)
\right|^2
d\bk .
\]
Suppose that $H_{\by}\in L^1(\mathbb R^2)$.
Then, in the strong-scattering regime
$|\sigma_V\Delta k|\to\infty$ and under the scaling
\begin{equation}\label{eq:gs}
\frac{\ell_V\ell_\Psi\omega_0}
     {c_0 z_{\rm int}}
=
\gamma_*
(\sigma_V\Delta k),
\qquad
\gamma_*=O(1) .
\end{equation}
 We have
\[
{\rm Var}(
\mathcal I_{\rm mf}(\by)
)
=
\left(
{\mathcal N}_{\beta_{\rm int}}
*
|\Psi|^2
\right)(\by)
+o(1),
\]
where
\[
{\mathcal N}_{\beta}(\bx)
=
\frac{1}{2\pi\beta^2}
\exp\!\left(
-\frac{|\bx|^2}{2\beta^2}
\right) \qquad \text{and}\qquad  
 \beta_{\rm int}
=
z_{\rm int}
\left|
1-\frac{c_0}{c_1}
\right|
\sqrt{\mathcal D}.
\]
\end{proposition}

This result explains why the image statistics become universal in the
strong-scattering regime. Once the coherent field has been suppressed,
the detailed form of the interface covariance function no longer
matters. Only its local curvature at the origin, encoded by
$\mathcal D$, survives in the asymptotic limit. Consequently, the
effect of the rough interface is equivalent to a Gaussian blurring of
the source intensity with width $\beta_{\rm int}$.

 \begin{proof} 
 
Consider the second moment of the matched field imaging function, the first term in \eqref{eq:varmf}
\begin{align*}
     \E\left[  \CalI_{\rm mf}^2(\by)  \right]   & = \frac{1}{(2\pi)^2}
      \iint     e^{i \bk \br}     \exp^{- ( 1- \rho_V(|\br|/\ell_V)) (\sigma_V \Delta k)^2    }  
\left\vert \Psi\left(\by + \bk \frac{c_0 z_{\rm int}}{\omega_{\rm o}} \right)  \right\vert^2
    d\bk d\br 
   .
   \end{align*}
We write  here 
\begin{align*}
\Psi\left(\by\right)  =  \Psi'\left(\frac{\by}{\ell_\Psi}\right)  ,
\end{align*}
for $\ell_\Psi$ the characteristic support of $\Psi$ and then get via a change of variables 
(but retaining the variable notation) 
\begin{align*}
     \E\left[  \CalI_{\rm mf}^2(\ell_\Psi \by)  \right]   & = \frac{1}{(2\pi)^2}
      \iint     e^{i \bk \br}     e^{- \left( 1- \rho_V\left(|\br| \frac{c_0 z_{\rm int}}{\ell_V \ell_\Psi \omega_{\rm o} }\right)\right) (\sigma_V \Delta k)^2    }  
\left\vert \Psi'\left(\by  + \bk \right)  \right\vert^2
    d\bk d\br 
   .
   \end{align*}
We are interested in the critical regime with partial blurring and we accordingly
assume the scaling in  \eqref{eq:gs}.   We are furthermore  interested  in relatively
strong scattering, $|\sigma_V \Delta k| \gg 1$, so that the mean field in the observation plane is relatively small. 
We denote
\begin{align*}
 {\mathcal H}_\by(\br) = \frac{1}{(2\pi)^2}
 \iint   {e^{-i \bk \br}}  \left\vert \Psi'\left(\by  + \bk \right)  \right\vert^2
    d\bk  ,
\end{align*}
and  as stated assume $\mathcal{H}_{\by}\in L^1(\mathbb R^2)$. We then have 
\begin{align*}
     \E\left[  \CalI_{\rm mf}^2(\ell_\Psi \by)  \right]   & = 
      \int          e^{- \left( 1- \rho_V\left(\frac{|\br| }{ \gamma_* (\sigma_V \Delta k)  }\right)\right) (\sigma_V \Delta k)^2    }  
    {\mathcal H}_\by(\br)
    d\br 
   .
   \end{align*}
   Next, choose any $\D >0$ and $r_\D > 0 $ so that
\begin{align}\label{eq:Hb}
            \int_{|\br| > r_\D}     
    \vert {\mathcal H}_\by(\br) \vert
    d\br  < \D/3
   .
   \end{align}
Then 
  \begin{align*}
     \E\left[  \CalI_{\rm mf}^2(\ell_\Psi \by)  \right]   & = 
      \int_{|\br| \leq  r_\D}          e^{- \left( 1- \rho_V\left(\frac{|\br| }{ \gamma_* (\sigma_V \Delta k)  }\right)\right) (\sigma_V \Delta k)^2    }  
    {\mathcal H}_\by(\br)
    d\br  + \Delta I^{(1,\D)}(\by) 
   ,
   \end{align*}
with $|\Delta I^{(1,\D)}| < \D/3$. 
In view of Assumption \ref{ass1} we have
\[
\rho_V(r)
=
1-\frac{D}{2}r^2+o(r^2) ,
\qquad
D=-\rho_V''(0)>0 .
\]
Hence, we also have  for bounded $r$
\[
\left(
1-\rho_V
\!\left(
\frac{r}
     {\gamma_*(\sigma_V\Delta k)}
\right)
\right)
(\sigma_V\Delta k)^2
=
\frac{D r^2}{2\gamma_*^2}
+ h(r; \sigma_V\Delta k, \gamma_*) , 
\]
with 
\[
\lim_{s \to \infty} \sup_{0<r<r_\D} \vert h(r ;s,\gamma_*)\vert = 0 .
\]
 We can then write 
  \begin{align*}
     \E\left[  \CalI_{\rm mf}^2(\ell_\Psi \by)  \right]   & = 
      \int_{|\br| \leq  r_\D}          e^{-  \frac{D |\br|^2}{2 \gamma_*^2}    }  
      \left( 1+  
      \left(e^{{h}(|\br|; \sigma_V\Delta k, \gamma_*) } -1 \right)
      \right)
    {\mathcal H}_\by(\br)
    d\br  + \Delta I^{(1,\D)}(\by) 
   . 
   \end{align*}
We can then choose $s(\D,\gamma_*)>0 $ so that for $|\sigma_V \Delta k| > s(\D,\gamma_*)$ 
 \begin{align*}
     \E\left[  \CalI_{\rm mf}^2(\ell_\Psi \by)  \right]   & = 
      \int_{|\br| \leq  r_\D}          e^{-  \frac{D |\br|^2}{2 \gamma_*^2}    }  
      {\mathcal H}_\by(\br)
    d\br  + \Delta I^{(2,\D)}(\by) 
   ,
   \end{align*} 
with $|\Delta I^{(2,\D)}| < 2\D/3$.  We then also have in view of \eqref{eq:Hb}
\begin{align*}
     \E\left[  \CalI_{\rm mf}^2(\ell_\Psi \by)  \right]   & = 
      \int   e^{-  \frac{D |\br|^2}{2 \gamma_*^2}    }  
      {\mathcal H}_\by(\br)
    d\br  + \Delta I^{(3,\D)}(\by) 
   ,
   \end{align*} 
 with $|\Delta I^{(3,\D)}| <  \D$.  We can conclude 
\begin{align*}
     \E\left[  \CalI_{\rm mf}^2(\by)  \right]   & = \frac{1}{(2\pi)^2}
      \iint      {e^{-i \bk \br}}    e^{-  \frac{D |\br|^2}{2 (\gamma_* /\ell_\Psi)^2}    }  
\left\vert \Psi\left(\by + \bk  \right)  \right\vert^2
    d\bk d\br 
       \nonumber     \\ & \hbox{}  
     + \Delta I^{(3,\D)}(\by/\ell_\Psi) 
   .
   \end{align*}
Finally, note that
\begin{align*}
 \frac{\ell_\Psi \sqrt{D}}{\gamma_* } = \beta_{\rm int} ,
\end{align*}
with  $\beta_{\rm int}$ given in \eqref{def:Bm} and \eqref{eq:defB2} and we can write
under the scaling \eqref{eq:gs} 
\begin{align*}
   \lim_{|\sigma_V \Delta k| \to \infty }  \E\left[  \CalI_{\rm mf}^2(\by)  \right]   & =  
     \lim_{|\sigma_V \Delta k| \to \infty }  {\rm Var}\left[  \CalI_{\rm mf}^2(\by)  \right]  = 
     \left({\mathcal N}_{\beta_{\rm int}} \ast |\Psi|^2 \right)(\by)
      ,
   \end{align*}
     which  concludes the proof of the proposition. 
   
\end{proof}

\section{Fourth Moment of Scattering Operator}\label{app:lem}

In this appendix we derive Proposition~\ref{lem1}.  The result provides the asymptotic fourth-order statistics
of the scattering operator and is the key ingredient in the derivation of the
optical image covariance in Proposition~\ref{prop4}.

\begin{proposition}[Fourth-order statistics of the scattering operator]

Assume that $V$ satisfies Assumption~\ref{ass1}, let
$\gamma=1/2$, and consider the strong-scattering regime
\eqref{eq:regime}.
Then in this regime, as
$\omega_{\rm o}(\bcals_0-\bcals_1)\sigma_V\to\infty$,
the fourth-order moment of the scattering operator
$\mathfrak K^\e$ converges in the sense of tempered
distributions to
\begin{align}
&\E\Big[
\mathfrak K^\e(\tau,\bk,\bk')\,
\overline{\mathfrak K^\e(\tau,\bq,\bq')}\,
\overline{\mathfrak K^\e(a\tau,\tbk,\tbk')}\,
\mathfrak K^\e(a\tau,\tbq,\tbq')
\Big]
\nonumber\\
&\qquad =
\delta\!\big((\bk'-\bk)-(\bq'-\bq)\big)\,
\delta\!\big((\tbk'-\tbk)-(\tbq'-\tbq)\big)
\nonumber\\
&\qquad\quad\times
\delta\!\big(a(\bk'-\bk)-(\tbk'-\tbk)\big)\,
\frac{
\exp\!\left(
-\dfrac{ |\tbq'-\tbq|^2}
{2(a\tau)^2\CalD}
\right)}
{2\pi\tau^2\CalD},
\label{eq:4K_prop}
\end{align}
for any constant $a$, where
\[
\CalD=-\rho_V''(0)\sigma_V^2/\ell_V^2
\]
is the curvature parameter introduced in
Remark~\ref{remark1}.

\end{proposition}
 
\begin{remark}
The three delta functions in
\eqref{eq:4K_prop}  express conservation of transverse momentum shift (spatial frequency) across the four scattering events. In the strong-scattering limit only scattering configurations satisfying these conservation relations contribute at leading order. 
The remaining Gaussian factor describes the statistical
distribution of the momentum transfer.
Its variance is determined solely by the curvature parameter
$\CalD$ of the interface covariance function.
Consequently, in the strong-scattering regime the fourth-order
statistics become universal and depend on the rough interface
only through the  parameter $\CalD$.
%
 \end{remark}

\begin{proof}

Recall the definition of the scattering operator \eqref{eq:K},
\[
\mathfrak K^\e(\tau,\bk,\bk') = \frac{1}{(2\pi)^2}
\int 
e^{i(\bk'-\bk)\cdot\bx}
e^{i\tau V(\bx)}
d\bx .
\]
We consider the fourth-order moment
\begin{equation*}
M^\e
=
\E\Big[
\mathfrak K^\e(\tau,\bk,\bk')
\overline{\mathfrak K^\e(\tau,\bq,\bq')}
\overline{\mathfrak K^\e(a\tau,\tbk,\tbk')}
\mathfrak K^\e(a\tau,\tbq,\tbq')
\Big].
\end{equation*}
Substituting the definition of $\mathfrak K^\e$ gives
\begin{equation*}
M^\e
=
\frac{1}{(2\pi)^8}
\iiiint
e^{i\Phi}
\E\left[
e^{i\tau(
V(\bx)-V(\bx_1)
-aV(\bx_2)
+aV(\bx_3))}
\right]
d\bx d\bx_1 d\bx_2 d\bx_3 ,
\end{equation*}
where
\[
\Phi = (\bk'-\bk) \cdot \bx
-(\bq'-\bq) \cdot \bx_1
-(\tbk'-\tbk) \cdot \bx_2
+(\tbq'-\tbq) \cdot \bx_3 .
\]
Using stationarity of $V$ we introduce
\[
\by_1=\bx_1-\bx,
\qquad
\by_2=\bx_2-\bx,
\qquad
\by_3=\bx_3-\bx .
\]
Integration with respect to $\bx$ immediately yields
\begin{align}
M^\e
&=
\delta \Big(
(\bk'-\bk)
-(\bq'-\bq)
-(\tbk'-\tbk)
+(\tbq'-\tbq)
\Big)
\nonumber\\
&\qquad\times \frac{1}{(2\pi)^6}
\int
e^{-i(\bq'-\bq)\cdot\by_1}
e^{-i(\tbk'-\tbk)\cdot\by_2}
e^{i(\tbq'-\tbq)\cdot\by_3}
\nonumber\\
&\qquad\times
\E\left[
e^{i\tau(
V({\bf 0})
-V(\by_1)
-aV(\by_2)
+aV(\by_3))}
\right]
d\by_1 d\by_2 d\by_3 .
\label{eq:Me}
\end{align}
The strong-scattering regime implies that the dominant contribution comes from
configurations for which
$
\by_1,\
\by_2, \hbox{and}\,  \by_3
$
are close to the origin. Using Assumption~\ref{ass1} and the same local
quadratic expansion employed in Appendix~\ref{app:sp},
\[
C(0)
-\frac{\mathcal D}{2}|\bx|^2
+o(|\bx|^2),
\]
one finds
\begin{align*}
{\rm Var}\big(
V({\bf 0})
-V(\by_1)
-aV(\by_2)
&
+aV(\by_3)
\big)
\nonumber \\
&=
\mathcal D
\big|
\by_1 + a(\by_2-\by_3)
\big|^2
+
o(|\by|^2).
\end{align*}
Consequently, {as the interface fluctuations $V$ are assumed to be Gaussian}
\begin{align}
\E\left[
e^{i\tau(
V({\bf 0})
-V(\by_1)
-aV(\by_2)
+aV(\by_3))}
\right]
\simeq
\exp\left(
-\frac{\tau^2\mathcal D}{2}
\big|
\by_1 + a(\by_2-\by_3)
\big|^2
\right).
\label{eq:appDr2}
\end{align}
In view of \eqref{eq:appDr2} we introduce the new variables
\[
\bz = \by_1 + a(\by_2-\by_3),
\qquad
\bu=\by_2,
\qquad
\bv=\by_3 .
\]
Substituting these into \eqref{eq:Me}  gives
\begin{align*}
M^\e
&\sim
\delta\Big(
(\bk'-\bk)
-(\bq'-\bq)
-(\tbk'-\tbk)
+(\tbq'-\tbq)
\Big)
\nonumber\\
&\quad\times
\frac{1}{(2\pi)^6}
\int
e^{-i(\bq'-\bq)\cdot\bz}
e^{i\Lambda_u\cdot\bu}
e^{- i\Lambda_v\cdot\bv}
\exp\left(
-\frac{\tau^2\mathcal D}{2}
|\bz|^2
\right)
d\bz d\bu d\bv ,
\end{align*}
where
\[
\Lambda_u = a(\bq'-\bq) - (\tbk'-\tbk),
 \qquad 
 \Lambda_v  = a(\bq'-\bq) -(\tbq'-\tbq).
\]
Next integration with respect to $\bu$ and $\bv$ produces two  delta functions
and integration with respect to $\bz$ produces the Gaussian density, finally, 
 after recombining the three delta functions we then get 
  \begin{align*}
M^\e
&\sim
\delta\big(
(\bk'-\bk)-(\bq'-\bq)
\big)
\delta\big(
(\tbk'-\tbk)-(\tbq'-\tbq)
\big)
\nonumber\\
&\qquad\times
\delta\big(
a(\bk'-\bk)-(\tbk'-\tbk)
\big)
\frac{
\exp\left(
-\dfrac{|\tbq'-\tbq|^2}
{2(a\tau)^2\mathcal D}
\right)
}
{2\pi\tau^2\mathcal D},
\end{align*}
which is precisely \eqref{eq:4K}.
  
\end{proof}
      
 \end{document}